\documentclass[sigconf,authorversion,nonacm]{aamas}

\usepackage{balance}
\usepackage{amsmath,amsthm}
\theoremstyle{plain}
\newtheorem{theorem}{Theorem}
\newtheorem{lemma}{Lemma}
\newtheorem{proposition}{Proposition}

\theoremstyle{definition}
\newtheorem{definition}{Definition}
\theoremstyle{remark}

\usepackage{svg}
\usepackage{xspace}
\usepackage{proof}
\usepackage[linesnumbered,lined,ruled,noend]{algorithm2e}
\SetKwInOut{Input}{Input}\SetKwInOut{Output}{Output}
\usepackage[labelformat=simple]{subcaption}

\usepackage{wrapfig}
\usepackage{tikz}
  \tikzset{s/.style={circle,draw},t/.style={->,>=stealth}}
  \usetikzlibrary{shapes.geometric, arrows.meta, positioning, calc}
\usepackage{interval}
\usepackage{orcidlink}

\setcopyright{none}

\title{Compositional Shielding and Reinforcement Learning for Multi-Agent Systems}

\author{Asger Horn Brorholt \orcidlink{0009-0007-1824-0554}}
\affiliation{
  \institution{Aalborg University}
  \city{Aalborg}
  \country{Denmark}}
\email{asgerhb@cs.aau.dk}

\author{Kim Guldstrand Larsen \orcidlink{0000-0002-5953-3384}}
\affiliation{
  \institution{Aalborg University}
  \city{Aalborg}
  \country{Denmark}}
\email{kgl@cs.aau.dk}

\author{Christian Schilling \orcidlink{0000-0003-3658-1065}}
\affiliation{
  \institution{Aalborg University}
  \city{Aalborg}
  \country{Denmark}}
\email{christianms@cs.aau.dk}

\begin{abstract}
Deep reinforcement learning has emerged as a powerful tool for obtaining high-performance policies. However, the safety of these policies has been a long-standing issue. One promising paradigm to guarantee safety is a \emph{shield}, which ``shields'' a policy from making unsafe actions. However, computing a shield scales exponentially in the number of state variables. This is a particular concern in multi-agent systems with many agents.
In this work, we propose a novel approach for multi-agent shielding. We address scalability by computing individual shields for each agent. The challenge is that typical safety specifications are global properties, but the shields of individual agents only ensure local properties. Our key to overcome this challenge is to apply assume-guarantee reasoning. Specifically, we present a sound proof rule that decomposes a (global, complex) safety specification into (local, simple) obligations for the shields of the individual agents.
Moreover, we show that applying the shields during reinforcement learning significantly improves the quality of the policies obtained for a given training budget. 
We demonstrate the effectiveness and scalability of our multi-agent shielding framework in two case studies, reducing the computation time from hours to seconds and achieving fast learning convergence.
\end{abstract}


\newcommand{\BibTeX}{\rm B\kern-.05em{\sc i\kern-.025em b}\kern-.08em\TeX}

\newcommand{\uppaalstratego}{\textsc{Uppaal Stratego}\xspace}
\newcommand{\states}{\ensuremath{\mathit{S}}\xspace}
\newcommand{\act}{\ensuremath{\mathit{Act}}\xspace}
\newcommand{\trans}{\ensuremath{\mathit{T}}\xspace}
\newcommand{\R}{\mathbb R}
\newcommand{\mdp}{\mathcal{M}}
\newcommand{\enabled}{\ensuremath{\mathcal{E}}\xspace}
\newcommand{\strategy}{\ensuremath{\sigma}\xspace}
\newcommand{\policy}{\ensuremath{\pi}\xspace}

\newcommand{\localshield}{\ensuremath{\ensuremath{\text{\protect 
{\color{white} $\nabla$} 
\hspace{-2.5ex}
\tikz[baseline]{
    \draw [line width=0.12ex]
        (0, 1.2ex) -- 
        (1.2ex, 1.2ex) --
        (1.2ex, 0.4ex) --
        (0.6ex, -0.25ex) --
        (0, 0.4ex) --
        cycle
}
\hspace{-1.1ex}}}\xspace}\xspace}
\newcommand{\ext}[1]{\ensuremath{{\uparrow}(#1)}\xspace}
\newcommand{\powerset}[1]{\ensuremath{2^{#1}}\xspace}
\newcommand{\lts}{\ensuremath{\mathcal{T}}\xspace}
\newcommand{\prj}{\ensuremath{\mathit{prj}}\xspace}
\newcommand{\cprj}{\ensuremath{\overline{\prj}}\xspace}
\newcommand{\sandbox}{\ensuremath{\mathit{sandbox}}\xspace}
\newcommand{\AG}[3]{\ensuremath{\langle #1 \rangle #2 \langle #3 \rangle}\xspace}
\newcommand{\comp}{\ensuremath{\sqcap}\xspace}
\newcommand{\prob}{\ensuremath{\mathit{Pr}}\xspace}
\definecolor{turquoise}{HTML}{1ABC9C}
\definecolor{emerald}{HTML}{2ECC71}
\definecolor{peterriver}{HTML}{3498DB}
\definecolor{amethyst}{HTML}{9B59B6}
\definecolor{wetasphalt}{HTML}{34495E}
\definecolor{greensea}{HTML}{16A085}
\definecolor{nephritis}{HTML}{27AE60}
\definecolor{belizehole}{HTML}{2980B9}
\definecolor{wisteria}{HTML}{8E44AD}
\definecolor{midnightblue}{HTML}{2C3E50}
\definecolor{sunflower}{HTML}{F1C40F}
\definecolor{carrot}{HTML}{E67E22}
\definecolor{alizarin}{HTML}{E74C3C}
\definecolor{clouds}{HTML}{ECF0F1}
\definecolor{concrete}{HTML}{95A5A6}
\definecolor{orange}{HTML}{F39C12}
\definecolor{pumpkin}{HTML}{D35400}
\definecolor{pomegranate}{HTML}{C0392B}
\definecolor{silver}{HTML}{BDC3C7}
\definecolor{asbestos}{HTML}{7F8C8D}

\newif\ifshowappendix
\showappendixtrue
\newcommand{\refappendix}[1]{\ifshowappendix{ #1}\else{~\cite[appendix]{BLS24}}\fi}

\begin{document}

\pagestyle{fancy}
\fancyhead{}

\maketitle

\section{Introduction}

Reinforcement learning (RL)~\cite{DBLP:books/wi/Puterman94,DBLP:books/lib/SuttonB98}, and in particular deep RL, has demonstrated success in automatically learning high-performance policies for complex systems~\cite{DBLP:journals/nature/MnihKSRVBGRFOPB15,DBLP:journals/nature/BellemareCCGMMP20}.
However, learned policies lack guarantees, which prevents applications in safety-critical domains.

An attractive algorithmic paradigm to provably safe RL is \emph{shielding}~\cite{DBLP:conf/aaai/AlshiekhBEKNT18}.
In this paradigm, one constructs a \emph{shield}, which is a nondeterministic policy that only allows safe actions.
The shield acts as a guardrail for the RL agent to enforce safety both during learning (of a concrete policy) and operation.
This way, one obtains a \emph{safe-by-design shielded policy with high performance}.

\emph{Shield synthesis} automatically computes a shield from a safety specification and a model of the system, but scales exponentially in the number of state variables.
This is a particular concern in multi-agent (MA) systems, which typically consist of many variables.
Shielding of MA systems will be our focus in this work.

Existing approaches to MA shielding address scalability by computing individual shields for each agent.
Yet, these shields are either not truly safe or not truly independent; rather, they require online communication among all agents, which is often unrealistic.

In this paper, we present the first MA shielding approach that is truly compositional, does not require online communication, and provides absolute safety guarantees.
Concretely, we assume that agents observe a subset of all system variables (i.e., operate in a projection of the global state space).
We show how to tractably synthesize individual shields in low-dimensional projections.
The challenge we need to overcome is that a straightforward generalization of the classical shield synthesis to the MA setting for truly independent shields often fails.
The reason is that the projection removes the potential to coordinate between the agents, but often some form of coordination is required.

To address the need for coordination, we get inspiration from \emph{compositional reasoning}, which is a powerful approach, allowing to scale up the analysis of distributed systems.
The underlying principle is to construct a correctness proof of multi-component systems by smaller, ``local'' proofs for each individual component.
In particular, \emph{assume-guarantee reasoning} for concurrent programs was popularized in seminal works~\cite{OwickiG76,Lamport77,Pnueli84,Stark85,DBLP:journals/fteda/BenvenisteCNPRR18}.
By writing $\AG{A}{C}{G}$ for ``assuming~$A$, component~$C$ will guarantee~$G$,'' the standard (acyclic) assume-guarantee rule for finite state machines with handshake synchronization looks as follows~\cite{DBLP:reference/mc/GiannakopoulouNP18}:
\[
  \infer{\AG{\top}{C_1 \Vert C_2 \Vert \cdots \Vert C_n}{\phi}}{%
    \AG{\top}{C_1}{G_1},
    \AG{G_1}{C_2}{G_2},
    \dots,
    \AG{G_{n-2}}{C_{n-1}}{G_{n-1}},
    \AG{G_{n-1}}{C_n}{\phi}
  }
\]

By this chain of assume-guarantee pairs, it is clear that, together, the components ensure safety property~$\phi$.

In this work, we adapt the above rule to multi-agent shielding.
Instead of one shield for the whole system, we synthesize an individual shield for each agent, which together we call a \emph{distributed shield}.
Thus, we arrive at $n$ shield synthesis problems (corresponding to the rule's premise), but each of which is efficient.
In our case studies, this reduces the synthesis time from hours to seconds.
The \emph{guarantees~$G_i$ allow the individual shields to coordinate} on responsibilities at synthesis time.
Yet, distributed shields do not require communication when deployed.
Altogether, this allows us to \emph{synthesize safe shields} in a compositional  and scalable way.

The crucial challenge is that, in the classical setting, the components~$C_i$ are fixed.
In our synthesis setting, the components~$C_i$ are our agents, which are \emph{not} fixed at the time of the shield synthesis.
In this work, we assume that the guarantees~$G_i$ are given, which allows us to derive corresponding individual agent shields via standard shield synthesis.

\paragraph{Motivating Example}\label{sect:platoon}

A multi-agent car platoon with adaptive cruise controls consists of $n$ cars, numbered from back to front~\cite{DBLP:conf/birthday/LarsenMT15} (Figure~\ref{fig:platoon}).
The cars 1 to~$n-1$ are each controlled by an agent, while (front) car~$n$ is driven by the environment.
The state variables are the car velocities~$v_i$ and distances~$d_i$ between cars~$i$ and~$i+1$.
For $i < n$, car~$i$ follows car~$i+1$, observing the variables $(v_i, v_{i+1}, d_i)$.
With a decision period of 1~second, cars act by choosing an acceleration from
$\{-2, 0, 2\}$ [m/s$^2$].
Velocities are capped between $\interval{-10}{20}$\,m/s.

\begin{figure}[t]
  \centering
   \includegraphics[width=0.85\linewidth]{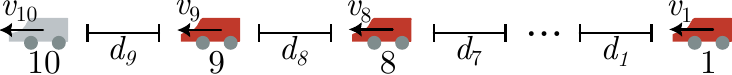}
  \caption{Car platoon example for $n = 10$ cars.}
  \Description{Cars driving in a row. A gray car labeled as number 10 is followed by red cars labeled 9, 8, down to 1. Car number 10 is shown to have velocity v10, and a distance d9 between itself and car 9. The other cars are similarly labeled.}
  \label{fig:platoon}
\end{figure}

When two cars have distance~0, they enter an uncontrollable ``damaged'' state where both cars get to a standstill.
The global safety property is to maintain a safe but bounded distance between all cars, i.e., the set of safe states is $\phi=\{ s \mid \bigwedge_i 0 < d_i < 200 \}$.

As a first attempt, we design the agents' individual safety properties to only maintain a safe distance to the car in front, i.e., $\phi_i=\{ s \mid 0 < d_i < 200 \}$ for all~$i$.
However, safe agent shields for cars~$i > 1$ do not exist for this property: car~$i$ cannot prevent a crash with car~$i-1$ (behind), and in the ``damaged'' (halting) state car~$i$ cannot guarantee to avoid crashing with car~$i+1$.  Note that making the distance $d_{i-1}$  observable for car~$i$ does not help.

To overcome this seemingly impossible situation, we will allow car~$i$ to \emph{assume} that the (unobservable) car~$i-1$ \emph{guarantees} to never crash into car~$i$.
This guarantee will be provided by the shield of car~$i-1$ and eliminates the critical behavior preventing a local shield for car~$i$.
In that way, we iteratively obtain local shields for all agents.
Note that this coincides with human driver reasoning.

\smallskip

Beside synthesis of a distributed shield, we also study learning policies for shielded agents.
In general, multi-agent reinforcement learning (MARL)~\cite{DBLP:journals/corr/ZhangYB19} is complex due to high-dimensional state and action spaces, which impede convergence to optimal policies.

Here, we identify a class of systems where learning the agents in a \emph{cascading} way is both effective and efficient.
Concretely, if we assign an index to each agent, and each agent only depends on agents with lower index, we can learn policies in a sequential order.
This leads to a low-dimensional space for the learning algorithm, which leads to fast convergence.
While in general suboptimal, we show that this approach still leads to Pareto-optimal results.

\smallskip

In summary, this paper makes the following main contributions:
\begin{itemize}
  \item We propose \emph{distributed shielding}, the first MA shielding approach with absolute safety guarantees and scalability, yet without online communication.
  To this end, our approach integrates shield synthesis and assume-guarantee reasoning.

  \item We propose (shielded) \emph{cascading learning}, a scalable MARL approach for systems with acyclic dependency structure, which further benefits from assume-guarantee reasoning.

  \item We evaluate our approaches in two case studies.
  First, we demonstrate that distributed shielding is scalable and, thanks to the integration of assume-guarantee reasoning, applicable.
  Second, we demonstrate that shielded cascading learning is efficient and achieves state-of-the-art performance.
\end{itemize}

\subsection{Related Work}

\paragraph{Shielding.}

As mentioned, shielding is a technique that computes a \emph{shield}, which prevents an agent from taking unsafe actions.
Thus, any policy under a shield is safe, which makes it attractive for safety both during learning and after deployment.
Shields are typically based on game-theoretic results, where they are called \emph{winning strategies}~\cite{DBLP:reference/mc/BloemCJ18}.
Early applications of shields in learning were proposed for timed systems~\cite{DBLP:conf/atva/DavidJLLLST14} and discrete systems~\cite{DBLP:conf/aaai/AlshiekhBEKNT18}.
The idea has since been extended to probabilistic systems~\cite{DBLP:conf/concur/0001KJSB20,DBLP:conf/ijcai/YangMRR23}, partial observability~\cite{DBLP:conf/aaai/Carr0JT23}, and continuous-time dynamics~\cite{DBLP:conf/aisola/BrorholtJLLS23,DBLP:conf/aisola/BrorholtHLS24}.
For more background we refer to surveys~\cite{DBLP:conf/birthday/KonighoferBEP22,DBLP:journals/tmlr/KrasowskiTM0WA23}.
In this work, we focus on discrete but multi-agent systems, which we now review in detail.

\paragraph{Multi-agent shielding.}

An early work on multi-agent enforcement considered a very restricted setting with deterministic environments where the specification is already given in terms of valid actions and not in terms of states~\cite{DBLP:conf/amcc/BharadwajBDKT19}.
Thus, the shield does not reason about the dynamics and simply overrides forbidden actions.

Model-predictive shielding assumes a backup policy together with a set of recoverable states from which this policy can guarantee safety.
Such a backup policy may for instance be implemented by a shield, and is combined with another (typically learned) policy.
First, a step with the second policy is simulated and, when the target state is recoverable, this step is executed; otherwise, the fallback policy is executed.
Crucially, this assumes that the environment is deterministic.
Zhang et al.\ proposed a multi-agent version~\cite{DBLP:journals/corr/ZhangB19}, where the key insight is that only some agents need to use the backup policy.
For scalability, the authors propose a greedy algorithm to identify a sufficiently small subset of agents.
However, the ``shield'' is centralized, which makes this approach not scalable.

Another work computes a safe policy online~\cite{RajuBDT21}, which may be slow.
Agents in close proximity create a communication group, and they communicate their planned trajectories for the next $k$ steps.
Each agent has an agreed-on priority in which they have to resolve safety violations, but if that is not possible, agents may disturb higher-priority agents.
The approach requires strong assumptions like deterministic system dynamics and immediate communication.

One work suggests to directly reinforcement-learn policies by simply encouraging safety~\cite{DBLP:conf/iclr/QinZCCF21}.
Here, the loss function encodes a safety proof called \emph{barrier certificate}.
But, as with any reward engineering, this approach does not guarantee safety in any way.

Another way to scale up shielding for multi-agent systems is a so-called \emph{factored shield}, which safeguards only a subset of the state space, independent of the number of agents~\cite{DBLP:conf/atal/Elsayed-AlyBAET21}.
When an agent moves, it joins or leaves a shield at border states.
However, this approach relies on very few agents ever interacting with each other, as otherwise, there is no significant scalability gain.

Factored shields were extended to \emph{dynamic shields}~\cite{DBLP:conf/atal/XiaoLD23}.
The idea is that, in order to reduce the communication overhead, an agent's shield should ``merge'' dynamically with the shields of other agents in the proximity.
Since the shields are computed with a $k$-step lookahead only, safety is not guaranteed invariantly.

\paragraph{Multi-agent verification.}

\emph{Rational verification} proposes to study specifications only from initial states in Nash equilibria, i.e., assuming that all agents act completely rationally~\cite{DBLP:journals/apin/AbateGHHKNPSW21}.
While that assumption may be useful for rational/optimal agents, we typically have learned agents in mind, which do not always act optimally.

The tool \emph{Verse} lets users specify multi-agent scenarios in a Python dialect and provides black-box (simulations) and white-box (formal proofs; our setting) analysis for time-bounded specifications~\cite{DBLP:conf/cav/LiZBSM23}.

Assume-guarantee reasoning has been applied to multi-agent systems in~\cite{DBLP:conf/amcc/PartoviL14} and in~\cite{DBLP:conf/prima/MikulskiJK22}, but not yet to (multi-agent) shielding.

\paragraph{Outline.}

In the next section, we define basic notation.
In Section~\ref{sect:shielding}, we introduce distributed shielding based on projections and extend it with assume-guarantee reasoning.
In Section~\ref{sect:learning}, we develop cascading learning, tailored to systems with acyclic dependencies.
In Section~\ref{sect:evaluation}, we evaluate our approaches in two case studies.
In Section~\ref{sect:conclusion}, we conclude and discuss future work.

\section{Preliminaries}\label{sect:preliminaries}

\subsection{Transition Systems (MDPs \& LTSs)}

We start with some basic definitions of transition systems.

\begin{definition}[Labeled transition system]
  A \emph{labeled transition system} (LTS) is a triple $\lts = (\states, \act, \trans)$ where $\states$ is the finite state space, $\act$ is the action space, and ${\trans} \subseteq \states \times \act \times \states$ is the transition relation with no dead ends, i.e., for all $s \in \states$ there exists some $a \in \act$ and $s' \in \states$ such that $(s, a, s') \in \trans$.
\end{definition}

\begin{definition}[Markov decision process]
  A \emph{Markov decision process} (MDP) is a triple $\mdp = (\states, \act, P)$ where $\states$ is the finite state space, $\act$ is the action space, and $P \colon \states \times \act \times \states \to [0, 1]$ is the probabilistic transition relation satisfying $\sum_{s' \in \states} P(s, a, s') \in \{0, 1\}$ for all $s \in \states$ and $a \in \act$, and for at least one action, the sum is~1.
\end{definition}

We will view an LTS as an abstraction of an MDP where probabilities are replaced by possibilities.

\begin{definition}[Induced LTS]
  Given an MDP $\mdp = (\states, \act, P)$, the \emph{induced LTS} is $\lts_\mdp = (\states, \act, \trans)$ with $(s, a, s') \in {\trans}$ iff $P(s, a, s') > 0$.
\end{definition}

\begin{definition}[Run]
  Assume an LTS $\lts = (\states, \act, \trans)$ and a finite alternating sequence of states and actions $\rho = s_0 a_0 s_1 a_1 \dots$;
  then, $\rho$ is a \emph{run} of $\lts$ if $(s_i, a_i, s_{i+1}) \in {\trans}$ for all $i \ge 0$.
  Similarly, for an MDP $\mdp = (\states, \act, P)$, $\rho$ is a \emph{run} of $\mdp$ if $P(s_i, a_i, s_{i+1}) > 0$ for all $i \ge 0$.
\end{definition}

We distinguish between strategies and policies in this work.
A strategy prescribes a nondeterministic choice of actions in each LTS state.
Similarly, a policy prescribes a probabilistic choice of actions in each MDP state.
Before defining them formally, we need a notion of restricting the actions to sensible choices.

\begin{definition}[Enabled actions]
  Given an LTS, $\enabled(s) = \{a \in \act \mid \exists s'\colon (s, a, s') \in {\trans}\}$ denotes the \emph{enabled actions} in state~$s$.
  Similarly, given an MDP,
  $\enabled(s) = \{a \in \act \mid \exists s'\colon P(s, a, s') > 0\}$.
\end{definition}

\begin{definition}[Strategy; policy]
  Given an LTS, a (nondeterministic) \emph{strategy} is a function $\strategy \colon \states \to \powerset{\act}$ such that $\emptyset \ne \strategy(s) \subseteq \enabled(s)$ for all $s \in \states$.
  Given an MDP, a (probabilistic) \emph{policy} is a function $\policy \colon \states \times \act \to [0, 1]$ such that $\sum_{a \in \enabled(s)} \policy(s, a) = 1$ and $\bigwedge_{a' \in \act \setminus \enabled(s)} \policy(s, a') = 0$ for all $s \in \states$.
\end{definition}

Note that our strategies and policies are memoryless.
This is justified as we will only consider safety properties in this work, for which memory is not required~\cite{DBLP:reference/mc/BloemCJ18}.
Strategies and policies restrict the possible runs, and we call these runs the outcomes.

\begin{definition}[Outcome]
  A run $\rho = s_0 a_0 s_1 a_1 \dots$ of an LTS is an \emph{outcome} of a strategy~$\strategy$ if $a_i \in \strategy(s_i)$ for all $i \ge 0$.
  Similarly, a run $\rho = s_0 a_0 s_1 a_1 \dots$ of an MDP is an outcome of a policy~$\policy$ if $\policy(s_i, a_i) > 0$ for all $i \ge 0$.
\end{definition}

\subsection{Safety and Shielding}

In this work, we are interested in safety properties, which are characterized by a set of safe (resp.\ unsafe) states.
The goal is to stay in the safe (resp.\ avoid the unsafe) states.
In this section, we introduce corresponding notions, in particular (classical) shields and how they can be applied.

\begin{definition}[Safety property]
  A \emph{safety property} is a set of states $\phi \subseteq \states$.
\end{definition}

\begin{definition}[Safe run]
  Given a safety property $\phi \subseteq \states$, a run $s_0 a_0 s_1 a_1 \dots$ is \emph{safe} if $s_i \in \phi$ for all $i \ge 0$.
\end{definition}

Given an LTS, a safety property $\phi \subseteq \states$ partitions the states into two sets: the \emph{winning states}, from which a strategy exists whose outcomes are all safe, and the complement.
The latter can be computed as the attractor set of the complement~$\states \setminus \phi$~\cite{DBLP:reference/mc/BloemCJ18}.
Since it is hopeless to ensure safe behavior from the complement states, in the following we will only be interested in outcomes starting in winning states, which we abstain from mentioning explicitly.

A shield is a (typically nondeterministic) strategy that ensures safety.
In game-theory terms, a shield is called a \emph{winning strategy}.

\begin{definition}[Shield]
  Given an LTS $(\states, \act, \trans)$ and a safety property $\phi \subseteq \states$, a \emph{shield} $\ensuremath{\text{\protect 
{\color{white} $\nabla$} 
\hspace{-2.5ex}
\tikz[baseline]{
    \draw [line width=0.12ex]
        (0, 1.2ex) -- 
        (1.2ex, 1.2ex) --
        (1.2ex, 0.4ex) --
        (0.6ex, -0.25ex) --
        (0, 0.4ex) --
        cycle
}
\hspace{-1.1ex}}}\xspace[\phi]$ is a strategy whose outcomes starting in any winning state are all safe wrt.~$\phi$.
\end{definition}

We often omit~$\phi$ and just write~$\ensuremath{\text{\protect 
{\color{white} $\nabla$} 
\hspace{-2.5ex}
\tikz[baseline]{
    \draw [line width=0.12ex]
        (0, 1.2ex) -- 
        (1.2ex, 1.2ex) --
        (1.2ex, 0.4ex) --
        (0.6ex, -0.25ex) --
        (0, 0.4ex) --
        cycle
}
\hspace{-1.1ex}}}\xspace$.
Among all shields, it is known that there is a ``best'' one that allows the most actions.

\begin{definition}[Most permissive shield]
  Given an LTS and a safety property~$\phi$, the \emph{most permissive shield} $\ensuremath{\text{\protect 
{\color{white} $\nabla$} 
\hspace{-2.5ex}
\tikz[baseline]{
    \draw [line width=0.12ex]
        (0, 1.2ex) -- 
        (1.2ex, 1.2ex) --
        (1.2ex, 0.4ex) --
        (0.6ex, -0.25ex) --
        (0, 0.4ex) --
        cycle
}
\hspace{-1.1ex}}}\xspace^*[\phi]$ is the shield that allows the largest set of actions for each state~$s \in \states$.
\end{definition}

\begin{lemma}[\cite{DBLP:reference/mc/BloemCJ18}]
  $\ensuremath{\text{\protect 
{\color{white} $\nabla$} 
\hspace{-2.5ex}
\tikz[baseline]{
    \draw [line width=0.12ex]
        (0, 1.2ex) -- 
        (1.2ex, 1.2ex) --
        (1.2ex, 0.4ex) --
        (0.6ex, -0.25ex) --
        (0, 0.4ex) --
        cycle
}
\hspace{-1.1ex}}}\xspace^*$ is unique and obtained as the union of all shields~$\ensuremath{\text{\protect 
{\color{white} $\nabla$} 
\hspace{-2.5ex}
\tikz[baseline]{
    \draw [line width=0.12ex]
        (0, 1.2ex) -- 
        (1.2ex, 1.2ex) --
        (1.2ex, 0.4ex) --
        (0.6ex, -0.25ex) --
        (0, 0.4ex) --
        cycle
}
\hspace{-1.1ex}}}\xspace$ for~$\phi$:
  $\ensuremath{\text{\protect 
{\color{white} $\nabla$} 
\hspace{-2.5ex}
\tikz[baseline]{
    \draw [line width=0.12ex]
        (0, 1.2ex) -- 
        (1.2ex, 1.2ex) --
        (1.2ex, 0.4ex) --
        (0.6ex, -0.25ex) --
        (0, 0.4ex) --
        cycle
}
\hspace{-1.1ex}}}\xspace^*(s) = \{a \in \act \mid \exists \ensuremath{\text{\protect 
{\color{white} $\nabla$} 
\hspace{-2.5ex}
\tikz[baseline]{
    \draw [line width=0.12ex]
        (0, 1.2ex) -- 
        (1.2ex, 1.2ex) --
        (1.2ex, 0.4ex) --
        (0.6ex, -0.25ex) --
        (0, 0.4ex) --
        cycle
}
\hspace{-1.1ex}}}\xspace\colon a \in \ensuremath{\text{\protect 
{\color{white} $\nabla$} 
\hspace{-2.5ex}
\tikz[baseline]{
    \draw [line width=0.12ex]
        (0, 1.2ex) -- 
        (1.2ex, 1.2ex) --
        (1.2ex, 0.4ex) --
        (0.6ex, -0.25ex) --
        (0, 0.4ex) --
        cycle
}
\hspace{-1.1ex}}}\xspace(s)\}$.
\end{lemma}

The standard usage of a shield is to restrict the actions of a policy for guaranteeing safety.
In this work, we also compose it with another strategy.
For that, we introduce the notion of composition of strategies (recall that a shield is also a strategy).
We can, however, only compose strategies that are compatible in the sense that they allow at least one common action in each state (otherwise the result is not a strategy according to our definition).

\begin{definition}[Composition]
  Two strategies~$\strategy_1$ and~$\strategy_2$ over an LTS $(\states, \act, \trans)$ are \emph{compatible} if $\strategy_1(s) \cap \strategy_2(s) \ne \emptyset$ for all $s \in \states$.

  Given compatible strategies~$\strategy$ and~$\strategy'$, their composition $\strategy \comp \strategy'$ is the strategy
  $(\strategy \comp \strategy')(s) = \strategy(s) \cap \strategy'(s)$.

  We write $\comp_{i<j} \; \strategy_i$ to denote $\strategy_1 \comp \dots \comp \strategy_{j-1}$, and $\comp_i \, \strategy_i$ to denote $\strategy_1 \comp \dots \comp \strategy_n$ when $n$ is clear from the context.
\end{definition}

Given a strategy~$\strategy$ and a compatible shield~$\ensuremath{\text{\protect 
{\color{white} $\nabla$} 
\hspace{-2.5ex}
\tikz[baseline]{
    \draw [line width=0.12ex]
        (0, 1.2ex) -- 
        (1.2ex, 1.2ex) --
        (1.2ex, 0.4ex) --
        (0.6ex, -0.25ex) --
        (0, 0.4ex) --
        cycle
}
\hspace{-1.1ex}}}\xspace$, we also use the alternative notation of the \emph{shielded strategy}~$\ensuremath{\text{\protect 
{\color{white} $\nabla$} 
\hspace{-2.5ex}
\tikz[baseline]{
    \draw [line width=0.12ex]
        (0, 1.2ex) -- 
        (1.2ex, 1.2ex) --
        (1.2ex, 0.4ex) --
        (0.6ex, -0.25ex) --
        (0, 0.4ex) --
        cycle
}
\hspace{-1.1ex}}}\xspace(\strategy) = \strategy \comp \ensuremath{\text{\protect 
{\color{white} $\nabla$} 
\hspace{-2.5ex}
\tikz[baseline]{
    \draw [line width=0.12ex]
        (0, 1.2ex) -- 
        (1.2ex, 1.2ex) --
        (1.2ex, 0.4ex) --
        (0.6ex, -0.25ex) --
        (0, 0.4ex) --
        cycle
}
\hspace{-1.1ex}}}\xspace$.

Given a set of states~$\phi$, we are interested whether an LTS ensures that we will stay in that set~$\phi$, independent of the strategy.

\begin{definition}
  Assume an LTS~$\lts$ and a set of states~$\phi$.
  We write $\lts \models \phi$ if for all strategies~$\strategy$, all corresponding outcomes $s_0 a_0 s_1 a_1 \dots$ satisfy $s_i \in \phi$ for all $i \ge 0$.
\end{definition}

We now use a different view on a shield and apply it to an LTS in order to ``filter out'' those actions that are forbidden by the shield.

\begin{definition}[Shielded LTS]
  Given an LTS $\lts = (\states, \act, \trans)$, a safety property~$\phi$, and a shield~$\ensuremath{\text{\protect 
{\color{white} $\nabla$} 
\hspace{-2.5ex}
\tikz[baseline]{
    \draw [line width=0.12ex]
        (0, 1.2ex) -- 
        (1.2ex, 1.2ex) --
        (1.2ex, 0.4ex) --
        (0.6ex, -0.25ex) --
        (0, 0.4ex) --
        cycle
}
\hspace{-1.1ex}}}\xspace[\phi]$, the \emph{shielded LTS} $\lts_\ensuremath{\text{\protect 
{\color{white} $\nabla$} 
\hspace{-2.5ex}
\tikz[baseline]{
    \draw [line width=0.12ex]
        (0, 1.2ex) -- 
        (1.2ex, 1.2ex) --
        (1.2ex, 0.4ex) --
        (0.6ex, -0.25ex) --
        (0, 0.4ex) --
        cycle
}
\hspace{-1.1ex}}}\xspace = (\states, \act, \trans_\ensuremath{\text{\protect 
{\color{white} $\nabla$} 
\hspace{-2.5ex}
\tikz[baseline]{
    \draw [line width=0.12ex]
        (0, 1.2ex) -- 
        (1.2ex, 1.2ex) --
        (1.2ex, 0.4ex) --
        (0.6ex, -0.25ex) --
        (0, 0.4ex) --
        cycle
}
\hspace{-1.1ex}}}\xspace)$ with ${\trans_\ensuremath{\text{\protect 
{\color{white} $\nabla$} 
\hspace{-2.5ex}
\tikz[baseline]{
    \draw [line width=0.12ex]
        (0, 1.2ex) -- 
        (1.2ex, 1.2ex) --
        (1.2ex, 0.4ex) --
        (0.6ex, -0.25ex) --
        (0, 0.4ex) --
        cycle
}
\hspace{-1.1ex}}}\xspace} = \{ (s, a, s') \in {\trans} \mid a \in \ensuremath{\text{\protect 
{\color{white} $\nabla$} 
\hspace{-2.5ex}
\tikz[baseline]{
    \draw [line width=0.12ex]
        (0, 1.2ex) -- 
        (1.2ex, 1.2ex) --
        (1.2ex, 0.4ex) --
        (0.6ex, -0.25ex) --
        (0, 0.4ex) --
        cycle
}
\hspace{-1.1ex}}}\xspace(s)\}$ is restricted to transitions whose actions are allowed by the shield.
\end{definition}

The next proposition asserts that a shielded LTS is safe.

\begin{proposition}
  Given an LTS~$\lts$, a safety property~$\phi$, and a corresponding shield~$\ensuremath{\text{\protect 
{\color{white} $\nabla$} 
\hspace{-2.5ex}
\tikz[baseline]{
    \draw [line width=0.12ex]
        (0, 1.2ex) -- 
        (1.2ex, 1.2ex) --
        (1.2ex, 0.4ex) --
        (0.6ex, -0.25ex) --
        (0, 0.4ex) --
        cycle
}
\hspace{-1.1ex}}}\xspace[\phi]$, all outcomes of any strategy for $\lts_\ensuremath{\text{\protect 
{\color{white} $\nabla$} 
\hspace{-2.5ex}
\tikz[baseline]{
    \draw [line width=0.12ex]
        (0, 1.2ex) -- 
        (1.2ex, 1.2ex) --
        (1.2ex, 0.4ex) --
        (0.6ex, -0.25ex) --
        (0, 0.4ex) --
        cycle
}
\hspace{-1.1ex}}}\xspace$ are safe.
\end{proposition}

In other words, $\lts_\ensuremath{\text{\protect 
{\color{white} $\nabla$} 
\hspace{-2.5ex}
\tikz[baseline]{
    \draw [line width=0.12ex]
        (0, 1.2ex) -- 
        (1.2ex, 1.2ex) --
        (1.2ex, 0.4ex) --
        (0.6ex, -0.25ex) --
        (0, 0.4ex) --
        cycle
}
\hspace{-1.1ex}}}\xspace \models \phi$.
We analogously define shielded MDPs.

\begin{definition}[Shielded MDP]
  Given an MDP $\mdp = (\states, \act, P)$, a safety property~$\phi$, and a shield~$\ensuremath{\text{\protect 
{\color{white} $\nabla$} 
\hspace{-2.5ex}
\tikz[baseline]{
    \draw [line width=0.12ex]
        (0, 1.2ex) -- 
        (1.2ex, 1.2ex) --
        (1.2ex, 0.4ex) --
        (0.6ex, -0.25ex) --
        (0, 0.4ex) --
        cycle
}
\hspace{-1.1ex}}}\xspace$ for~$\lts_\mdp$, the \emph{shielded MDP} $\mdp_\ensuremath{\text{\protect 
{\color{white} $\nabla$} 
\hspace{-2.5ex}
\tikz[baseline]{
    \draw [line width=0.12ex]
        (0, 1.2ex) -- 
        (1.2ex, 1.2ex) --
        (1.2ex, 0.4ex) --
        (0.6ex, -0.25ex) --
        (0, 0.4ex) --
        cycle
}
\hspace{-1.1ex}}}\xspace = (\states, \act, P_\ensuremath{\text{\protect 
{\color{white} $\nabla$} 
\hspace{-2.5ex}
\tikz[baseline]{
    \draw [line width=0.12ex]
        (0, 1.2ex) -- 
        (1.2ex, 1.2ex) --
        (1.2ex, 0.4ex) --
        (0.6ex, -0.25ex) --
        (0, 0.4ex) --
        cycle
}
\hspace{-1.1ex}}}\xspace)$ is restricted to transitions with actions allowed by~$\ensuremath{\text{\protect 
{\color{white} $\nabla$} 
\hspace{-2.5ex}
\tikz[baseline]{
    \draw [line width=0.12ex]
        (0, 1.2ex) -- 
        (1.2ex, 1.2ex) --
        (1.2ex, 0.4ex) --
        (0.6ex, -0.25ex) --
        (0, 0.4ex) --
        cycle
}
\hspace{-1.1ex}}}\xspace$:
  $P_\ensuremath{\text{\protect 
{\color{white} $\nabla$} 
\hspace{-2.5ex}
\tikz[baseline]{
    \draw [line width=0.12ex]
        (0, 1.2ex) -- 
        (1.2ex, 1.2ex) --
        (1.2ex, 0.4ex) --
        (0.6ex, -0.25ex) --
        (0, 0.4ex) --
        cycle
}
\hspace{-1.1ex}}}\xspace(s, a, s') = P(s, a, s')$ if $a \in \ensuremath{\text{\protect 
{\color{white} $\nabla$} 
\hspace{-2.5ex}
\tikz[baseline]{
    \draw [line width=0.12ex]
        (0, 1.2ex) -- 
        (1.2ex, 1.2ex) --
        (1.2ex, 0.4ex) --
        (0.6ex, -0.25ex) --
        (0, 0.4ex) --
        cycle
}
\hspace{-1.1ex}}}\xspace(s)$, and $P_\ensuremath{\text{\protect 
{\color{white} $\nabla$} 
\hspace{-2.5ex}
\tikz[baseline]{
    \draw [line width=0.12ex]
        (0, 1.2ex) -- 
        (1.2ex, 1.2ex) --
        (1.2ex, 0.4ex) --
        (0.6ex, -0.25ex) --
        (0, 0.4ex) --
        cycle
}
\hspace{-1.1ex}}}\xspace(s, a, s') = 0$ otherwise.
\end{definition}

\begin{proposition}
  Assume an MDP~$\mdp$, a safety property~$\phi$, and a corresponding shield~$\ensuremath{\text{\protect 
{\color{white} $\nabla$} 
\hspace{-2.5ex}
\tikz[baseline]{
    \draw [line width=0.12ex]
        (0, 1.2ex) -- 
        (1.2ex, 1.2ex) --
        (1.2ex, 0.4ex) --
        (0.6ex, -0.25ex) --
        (0, 0.4ex) --
        cycle
}
\hspace{-1.1ex}}}\xspace[\phi]$ for~$\lts_\mdp$.
  Then all outcomes of any policy for $\mdp_\ensuremath{\text{\protect 
{\color{white} $\nabla$} 
\hspace{-2.5ex}
\tikz[baseline]{
    \draw [line width=0.12ex]
        (0, 1.2ex) -- 
        (1.2ex, 1.2ex) --
        (1.2ex, 0.4ex) --
        (0.6ex, -0.25ex) --
        (0, 0.4ex) --
        cycle
}
\hspace{-1.1ex}}}\xspace$ are safe.
\end{proposition}

The last proposition explains how standard shielding is applied to learn safe policies.
Given an MDP~$\mdp$, we first compute a shield~$\ensuremath{\text{\protect 
{\color{white} $\nabla$} 
\hspace{-2.5ex}
\tikz[baseline]{
    \draw [line width=0.12ex]
        (0, 1.2ex) -- 
        (1.2ex, 1.2ex) --
        (1.2ex, 0.4ex) --
        (0.6ex, -0.25ex) --
        (0, 0.4ex) --
        cycle
}
\hspace{-1.1ex}}}\xspace$ over the induced LTS~$\lts_\mdp$.
Then we apply the shield to the MDP~$\mdp$ to obtain~$\mdp_\ensuremath{\text{\protect 
{\color{white} $\nabla$} 
\hspace{-2.5ex}
\tikz[baseline]{
    \draw [line width=0.12ex]
        (0, 1.2ex) -- 
        (1.2ex, 1.2ex) --
        (1.2ex, 0.4ex) --
        (0.6ex, -0.25ex) --
        (0, 0.4ex) --
        cycle
}
\hspace{-1.1ex}}}\xspace$ and filter unsafe actions.
The shield guarantees that the agent is safe both during and after learning.

From now on we mainly focus on computing shields from an LTS, as the generalization to MDPs is straightforward.

\subsection{Compositional Systems}

Now we turn to compositional systems (LTSs and MDPs) with multiple agents.
We restrict ourselves to $k$-dimensional state spaces~$\states$, i.e., products of variables $\states = \bigtimes_i \states_i$. We allow for sharing some of these variables among the agents by projecting to observation subspaces.
The following is the standard definition of projecting out certain variables while retaining others.
We use the notation that, given an $n$-vector $v = (v_1, \dots, v_n)$, $v[i]$ denotes the $i$-th element~$v_i$.

\begin{definition}[Projection]
  A \emph{projection} is a mapping $\prj \colon \states \to O$ that maps $k$-dimensional vectors $s \in \states$ to $j$-dimensional vectors $o \in O$, where $j \le k$.
  Formally, $\prj$ is associated with a sequence of $j$ indices $1 \le i_1 < \dots < i_j \le k$ such that $\prj(s) = (s[i_1], \dots, s[i_j])$.
  Additionally, we define $\prj(\phi) = \bigcup_{s \in \phi} \{\prj(s)\}$.
\end{definition}

\begin{definition}[Extension]
  Given projection $\prj \colon \states \to O$, the set of states projected to~$o$ is the \emph{extension} $\ext{o} = \{s \in \states \mid \prj(s) = o\}$.
\end{definition}

Later we will also use an alternative projection, which we call \emph{restricted}.
The motivation is that the standard projection above sometimes retains too many states.
The restricted projection instead only keeps those states such that the extension of the projection ($\ext{\cdot}$) is contained in the original set.
For instance, for the state space $\states = \{0, 1\}^2$, the set of states $\phi = \{(0, 0), (0, 1), (1, 0)\}$, and the one-dimensional projection $\prj(s) = s[1]$, we have that $\prj(\phi) = \{0, 1\}$.
The restricted projection removes $1$ as $(1, 1) \notin \phi$.

\begin{definition}[Restricted projection]
  A \emph{restricted projection} is a mapping $\cprj \colon 2^\states \to 2^O$ that maps sets of $k$-dimensional vectors $s \in \states$ to sets of $j$-dimensional vectors $o \in O$, where $j \le k$.
  Formally, $\cprj$ is associated with a sequence of $j$ indices $1 \le i_1 < \dots < i_j \le k$.
  Let $\prj$ be the corresponding (standard) projection and $\phi \subseteq \states$.
  Then $\cprj(\phi) = \{ o \in O \mid \{s \in \states \mid \prj(s) = o\} \subseteq \phi\}$.
  Again, we define $\cprj(\phi) = \bigcup_{s \in \phi} \{\cprj(s)\}$.
\end{definition}

We will apply $\cprj$ only to safety properties~$\phi$.
The following alternative characterization may help with the intuition:
$\cprj(\phi) = \overline{\prj(\overline{\phi})} = O \setminus \prj(\states \setminus \phi)$, where~$\overline{\phi}$ denotes the complement~$\states \setminus \phi$ (resp.\ $O \setminus \phi$) of a set of states~$\phi \subseteq \states$ (resp.\ observations $\phi \subseteq O$).

Crucially, $\prj$ and~$\cprj$ coincide if $\ext{\prj(\phi)} = \phi$, i.e., if the projection of~$\phi$ preserves correlations.
We will later turn our attention to agent safety properties, where this is commonly the case.

Now we can define a multi-agent LTS and MDP.

\begin{definition}[$n$-agent LTS/MDP]
  An \emph{$n$-agent LTS} $(\states, \act, \trans)$ or an \emph{$n$-agent MDP} $(\states, \act, P)$ have an $n$-dimensional action space $\act = \act_1 \times \dots \times \act_n$ and a family of $n$ projections $\prj_i$, $i = 1, \dots, n$.
  Each \emph{agent}~$i$ is associated with the projection $\prj_i \colon \states \to O_i$ from $\states$ to its \emph{observation space}~$O_i$.
\end{definition}

We note that the observation space introduces partial observability.
Obtaining optimal strategies/policies for partial observability is difficult and generally requires infinite memory~\cite{DBLP:journals/jcss/ChatterjeeCT16}.
Since this is impractical, we restrict ourselves to memoryless strategies/policies.

We can apply the projection function~$\prj$ to obtain a ``local'' LTS, modeling partial observability.

\begin{definition}[Projected LTS]
  For an $n$-agent LTS $\lts = (\states, \act, \trans)$ and an agent~$i$ with projection function $\prj_i \colon \states \to O_i$, the \emph{projected LTS to agent~$i$} is $\lts^i = (O_i, \act_i, \trans_i)$ where $\act_i = \{a[i] \mid a \in \act\}$ and ${\trans_i} = \{(\prj_i(s), a[i], \prj_i(s)') \mid (s, a, s') \in {\trans}\}$.
\end{definition}

\section{Distributed Shield Synthesis}\label{sect:shielding}

We now turn to shielding in a multi-agent setting.
The straightforward approach is to consider the full-dimensional system and compute a global shield.
This has, however, two issues.
First, a global shield assumes communication among the agents, which we generally do not want to assume.
Second, and more importantly, shield computation scales exponentially in the number of variables.

To address these issues, we instead compute \emph{local} shields, one for each agent.
A local shield still keeps its agent safe.
But since we only consider the agent's observation space, the shield does not require communication, and the computation is much cheaper.

\subsection{Projection-Based Shield Synthesis}

Rather than enforcing the global safety property, local shields will enforce agent-specific properties, which we characterize next.

\begin{definition}[$n$-agent safety property]
  Given an $n$-agent LTS or MDP with state space~$\states$, a safety property $\phi \subseteq \states$ is an \emph{$n$-agent safety property} if $\phi = \bigcap_{i=1}^n \phi_i$ consists of \emph{agent safety properties}~$\phi_i$ for each agent~$i$.
\end{definition}

Note that we can let $\phi_i = \phi$ for all~$i$, so this is not a restriction.
But typically we are interested in properties that can be accurately assessed in the agents' observation space (i.e., $\prj_i(\phi_i) = \cprj_i(\phi_i)$).

Next, we define a local shield of an agent, which, like the agent, operates in the observation space.

\begin{definition}[Local shield]
  Given an $n$-agent LTS $\lts = (\states, \act, \trans)$ with observation spaces~$O_i$ and an $n$-agent safety property $\phi = \bigcap_{i=1}^n \phi_i \subseteq \states$, let $\ensuremath{\text{\protect 
{\color{white} $\nabla$} 
\hspace{-2.5ex}
\tikz[baseline]{
    \draw [line width=0.12ex]
        (0, 1.2ex) -- 
        (1.2ex, 1.2ex) --
        (1.2ex, 0.4ex) --
        (0.6ex, -0.25ex) --
        (0, 0.4ex) --
        cycle
}
\hspace{-1.1ex}}}\xspace_i \colon O_i \to \powerset{\act_i}$ be a shield for~$\lts^i$
  wrt.~$\cprj_i(\phi_i)$, for some agent $i \in \{1, \dots, n\}$, i.e., $\lts^i \models_{\localshield_i} \cprj_i(\phi_i)$.
  We call~$\ensuremath{\text{\protect 
{\color{white} $\nabla$} 
\hspace{-2.5ex}
\tikz[baseline]{
    \draw [line width=0.12ex]
        (0, 1.2ex) -- 
        (1.2ex, 1.2ex) --
        (1.2ex, 0.4ex) --
        (0.6ex, -0.25ex) --
        (0, 0.4ex) --
        cycle
}
\hspace{-1.1ex}}}\xspace_i$ a \emph{local shield} of agent~$i$.
\end{definition}

We define an operation to turn a $j$-dimensional (local) shield into a $k$-dimensional (global) shield.
This global shield allows all global actions whose projections are allowed by the local shield.

\begin{definition}[Extended shield]
  Assume an $n$-agent LTS $\lts = (\states, \act, \trans)$ with projections~$\prj_i$, an $n$-agent safety property $\phi = \bigcap_{i=1}^n \phi_i \subseteq \states$, and a corresponding local shield $\localshield_i$.
  The \emph{extended shield}~$\ext{\localshield_i}$ is defined as $\ext{\localshield_i}(s) = \{a \in \act \mid a[i] \in \localshield_i(\prj_i(s))\}$.
\end{definition}

The following definition is just syntactic sugar to ease reading.

\begin{definition}\label{def:models_shield}
  Assume an LTS $\lts$, a set of states $\phi$, and a shield~$\ensuremath{\text{\protect 
{\color{white} $\nabla$} 
\hspace{-2.5ex}
\tikz[baseline]{
    \draw [line width=0.12ex]
        (0, 1.2ex) -- 
        (1.2ex, 1.2ex) --
        (1.2ex, 0.4ex) --
        (0.6ex, -0.25ex) --
        (0, 0.4ex) --
        cycle
}
\hspace{-1.1ex}}}\xspace$ for~$\phi$.
  We write $\lts \models_\ensuremath{\text{\protect 
{\color{white} $\nabla$} 
\hspace{-2.5ex}
\tikz[baseline]{
    \draw [line width=0.12ex]
        (0, 1.2ex) -- 
        (1.2ex, 1.2ex) --
        (1.2ex, 0.4ex) --
        (0.6ex, -0.25ex) --
        (0, 0.4ex) --
        cycle
}
\hspace{-1.1ex}}}\xspace \phi$ as an alternative to $\lts_\ensuremath{\text{\protect 
{\color{white} $\nabla$} 
\hspace{-2.5ex}
\tikz[baseline]{
    \draw [line width=0.12ex]
        (0, 1.2ex) -- 
        (1.2ex, 1.2ex) --
        (1.2ex, 0.4ex) --
        (0.6ex, -0.25ex) --
        (0, 0.4ex) --
        cycle
}
\hspace{-1.1ex}}}\xspace \models \phi$.
\end{definition}

The following lemma says that it is sufficient to have a local shield ensuring the \emph{restricted} projection $\cprj_i(\phi_i)$ of an agent safety property~$\phi_i$ in order to guarantee safety of the extended shield.

\begin{lemma}\label{lem:proj}
  Assume an $n$-agent LTS $\lts$, a safety property~$\phi_i$, and a local shield $\localshield_i$ such that $\lts^i \models_{\localshield_i} \cprj_i(\phi_i)$.
  Then $\lts \models_{\ext{\localshield_i}} \phi_i$.
\end{lemma}

\begin{proof}
  The proof is by contraposition.
  Assume that there is an unsafe outcome $\rho$ in $\lts$ (starting in a winning state) under the extended shield~$\ext{\localshield_i}$, i.e., $\rho$ contains a state $s \notin \phi$.
  Then the projected run $\prj_i(s_0) \, a[i] \, \prj_i(s_1) \dots$ is an outcome of $\lts^i$ under local shield~$\localshield_i$, and $\prj_i(s) \notin \cprj_i(\phi)$ by the definition of $\cprj$.
  This contradicts that~$\localshield_i$ is a local shield.
\end{proof}

The following example shows that the \emph{restricted} projection is necessary.
Consider the LTS~$\lts$ where $\states = \{0, 1\}^2$, $\act = \{z, p\}^2$, and $\trans = \{$%
$((0,0), (z,z), (0,0)),$
$((0,0), (z,p), (0,1)),$
$((0,0), (p,z), (1,0)),$
$((0,0), (p,p), (1,1))$%
$\}$.
For $i = 1, 2$ let $\phi_i = \{(0,0),(0,1),(1,0)\}$ and $\prj_i$ project to the $i$-th component~$O_i$.
Then $\prj_i(\phi_i) = \{0, 1\} = \prj_i(\states)$, i.e., all states in the projection are safe, and hence a local shield may allow $\ensuremath{\text{\protect 
{\color{white} $\nabla$} 
\hspace{-2.5ex}
\tikz[baseline]{
    \draw [line width=0.12ex]
        (0, 1.2ex) -- 
        (1.2ex, 1.2ex) --
        (1.2ex, 0.4ex) --
        (0.6ex, -0.25ex) --
        (0, 0.4ex) --
        cycle
}
\hspace{-1.1ex}}}\xspace_i(0) = \{z, p\}$.
But then the unsafe state $(1, 1)$ would be reachable in~$\lts$.

\smallskip

If $\ensuremath{\text{\protect 
{\color{white} $\nabla$} 
\hspace{-2.5ex}
\tikz[baseline]{
    \draw [line width=0.12ex]
        (0, 1.2ex) -- 
        (1.2ex, 1.2ex) --
        (1.2ex, 0.4ex) --
        (0.6ex, -0.25ex) --
        (0, 0.4ex) --
        cycle
}
\hspace{-1.1ex}}}\xspace = \comp_i \, \ext{\localshield_i}$ exists, we call it a \emph{distributed shield}.
This terminology is justified in the next theorem, which says that we can synthesize $n$ local shields in the projections and then combine these local shields to obtain a safe shield for the global system.

\begin{theorem}[Projection-based shield synthesis]\label{thm:shield_simple}
  Assume an $n$-agent LTS $\lts = (\states, \act, \trans)$ and an $n$-agent safety property $\phi = \bigcap_{i=1}^n \phi_i \subseteq \states$.
  Moreover, assume local shields~$\ensuremath{\text{\protect 
{\color{white} $\nabla$} 
\hspace{-2.5ex}
\tikz[baseline]{
    \draw [line width=0.12ex]
        (0, 1.2ex) -- 
        (1.2ex, 1.2ex) --
        (1.2ex, 0.4ex) --
        (0.6ex, -0.25ex) --
        (0, 0.4ex) --
        cycle
}
\hspace{-1.1ex}}}\xspace_i$
  for all $i = 1, \dots, n$.
  If $\ensuremath{\text{\protect 
{\color{white} $\nabla$} 
\hspace{-2.5ex}
\tikz[baseline]{
    \draw [line width=0.12ex]
        (0, 1.2ex) -- 
        (1.2ex, 1.2ex) --
        (1.2ex, 0.4ex) --
        (0.6ex, -0.25ex) --
        (0, 0.4ex) --
        cycle
}
\hspace{-1.1ex}}}\xspace = \comp_i \, \ext{\localshield_i}$ exists, then~$\ensuremath{\text{\protect 
{\color{white} $\nabla$} 
\hspace{-2.5ex}
\tikz[baseline]{
    \draw [line width=0.12ex]
        (0, 1.2ex) -- 
        (1.2ex, 1.2ex) --
        (1.2ex, 0.4ex) --
        (0.6ex, -0.25ex) --
        (0, 0.4ex) --
        cycle
}
\hspace{-1.1ex}}}\xspace$ is a shield for $\lts$ wrt.\ $\phi$ (i.e., $\lts_\ensuremath{\text{\protect 
{\color{white} $\nabla$} 
\hspace{-2.5ex}
\tikz[baseline]{
    \draw [line width=0.12ex]
        (0, 1.2ex) -- 
        (1.2ex, 1.2ex) --
        (1.2ex, 0.4ex) --
        (0.6ex, -0.25ex) --
        (0, 0.4ex) --
        cycle
}
\hspace{-1.1ex}}}\xspace \models \phi$).
\end{theorem}

\begin{proof}
  By definition, each local shield~$\ensuremath{\text{\protect 
{\color{white} $\nabla$} 
\hspace{-2.5ex}
\tikz[baseline]{
    \draw [line width=0.12ex]
        (0, 1.2ex) -- 
        (1.2ex, 1.2ex) --
        (1.2ex, 0.4ex) --
        (0.6ex, -0.25ex) --
        (0, 0.4ex) --
        cycle
}
\hspace{-1.1ex}}}\xspace_i$ ensures that the (\emph{restricted} projected) agent safety property $\phi_i$ holds in $\lts^i$.
  Since $\lts^i$ is a projection of $\lts$, any distributed shield with $i$-th component $\ensuremath{\text{\protect 
{\color{white} $\nabla$} 
\hspace{-2.5ex}
\tikz[baseline]{
    \draw [line width=0.12ex]
        (0, 1.2ex) -- 
        (1.2ex, 1.2ex) --
        (1.2ex, 0.4ex) --
        (0.6ex, -0.25ex) --
        (0, 0.4ex) --
        cycle
}
\hspace{-1.1ex}}}\xspace_i$ also preserves $\phi_i$ in $\lts$ (by Lemma~\ref{lem:proj}).
  Hence, $\ensuremath{\text{\protect 
{\color{white} $\nabla$} 
\hspace{-2.5ex}
\tikz[baseline]{
    \draw [line width=0.12ex]
        (0, 1.2ex) -- 
        (1.2ex, 1.2ex) --
        (1.2ex, 0.4ex) --
        (0.6ex, -0.25ex) --
        (0, 0.4ex) --
        cycle
}
\hspace{-1.1ex}}}\xspace = \comp_i \, \ext{\localshield_i}$ ensures all agent safety properties $\phi_i$ and thus $\phi = \bigcap_{i=1}^n \phi_i$.
\end{proof}

\smallskip

Unfortunately, the theorem is often not useful in practice because the local shields may not exist.
The projection generally removes the possibility to coordinate with other agents.
By \emph{coordination} we do not mean (online) communication but simply (offline) agreement on ``who does what.''
Often, this coordination is necessary to achieve agent safety.
We address this lack of coordination in the next section.

\subsection{Assume-Guarantee Shield Synthesis}

Shielding an LTS removes some transitions.
Thus, by repeatedly applying multiple shields to the same LTS, we obtain a sequence of more and more restricted LTSs.

\begin{definition}[Restricted LTS]
  Assume two LTSs $\lts = (\states, \act, \trans)$, $\lts' = (\states, \act, \trans')$.
  We write $\lts \preceq \lts'$ if ${\trans} \subseteq {\trans'}$.
\end{definition}

\begin{lemma}\label{lem:preceq}
  Let $\lts \preceq \lts'$ be two LTSs.
  Then $\lts' \models \phi \implies \lts \models \phi$.
\end{lemma}

\begin{proof}
  As $\trans'$ contains all transitions of $\trans$, it has at least the same outcomes.
  If no outcome of~$\lts'$ leaves~$\phi$, the same holds for~$\lts$.
\end{proof}

We now turn to the main contribution of this section.
For a safety property $\phi'$, we assume an $n$-agent safety property $\phi = \bigcap_{i=1}^n \phi_i$ is given such that $\phi \subseteq \phi'$ (i.e., $\phi$ is more restrictive).
We use these agent safety properties~$\phi_i$ to filter out behavior during shield synthesis.
They may contain additional guarantees, which are used to coordinate responsibilities between agents.

Crucially, in our work, the guarantees are given in a certain order.
We assume \emph{wlog} that the agent indices are ordered from~1 to~$n$ such that agent~$i$ can only rely on the safety properties of all agents~$j < i$.
Thus, agent~$i$ guarantees~$\phi_i$ by assuming $\bigcap_{j < i}\phi_j$.
This is important to avoid problems with (generally unsound) circular reasoning.
In particular, agent~$1$ cannot rely on anything, and $\phi_n$ is not relied on.

The theorem then states that if each agent guarantees its safety property~$\phi_i$, and only relies on guarantees $\phi_j$ such that $j<i$. The result is a (safe) distributed shield.
The described condition is formally expressed as $\left(\lts_{\ensuremath{\text{\protect 
{\color{white} $\nabla$} 
\hspace{-2.5ex}
\tikz[baseline]{
    \draw [line width=0.12ex]
        (0, 1.2ex) -- 
        (1.2ex, 1.2ex) --
        (1.2ex, 0.4ex) --
        (0.6ex, -0.25ex) --
        (0, 0.4ex) --
        cycle
}
\hspace{-1.1ex}}}\xspace^*[\bigcap_{j < i} \phi_j]} \right)^i \models_{\localshield_i} \cprj_i(\phi_i)$, where we use the most permissive shield~$\ensuremath{\text{\protect 
{\color{white} $\nabla$} 
\hspace{-2.5ex}
\tikz[baseline]{
    \draw [line width=0.12ex]
        (0, 1.2ex) -- 
        (1.2ex, 1.2ex) --
        (1.2ex, 0.4ex) --
        (0.6ex, -0.25ex) --
        (0, 0.4ex) --
        cycle
}
\hspace{-1.1ex}}}\xspace^*$ for unicity.

\begin{theorem}[Assume-guarantee shield synthesis]\label{thm:shield_agr}
  Assume an $n$-agent LTS $\lts = (\states, \act, \trans)$ with projections $\prj_i$ and an $n$-agent safety property $\phi = \bigcap_i \phi_i$.
  Moreover, assume (local) shields $\localshield_i$ for all $i$ such that $\left(\lts_{\ensuremath{\text{\protect 
{\color{white} $\nabla$} 
\hspace{-2.5ex}
\tikz[baseline]{
    \draw [line width=0.12ex]
        (0, 1.2ex) -- 
        (1.2ex, 1.2ex) --
        (1.2ex, 0.4ex) --
        (0.6ex, -0.25ex) --
        (0, 0.4ex) --
        cycle
}
\hspace{-1.1ex}}}\xspace^*[\bigcap_{j < i} \phi_j]} \right)^i \models_{\localshield_i} \cprj_i(\phi_i)$.
  Then, if $\ensuremath{\text{\protect 
{\color{white} $\nabla$} 
\hspace{-2.5ex}
\tikz[baseline]{
    \draw [line width=0.12ex]
        (0, 1.2ex) -- 
        (1.2ex, 1.2ex) --
        (1.2ex, 0.4ex) --
        (0.6ex, -0.25ex) --
        (0, 0.4ex) --
        cycle
}
\hspace{-1.1ex}}}\xspace = \comp_i \, \ext{\localshield_i}$ exists, it is a shield for~$\lts$ wrt.~$\phi$ (i.e., $\lts_\ensuremath{\text{\protect 
{\color{white} $\nabla$} 
\hspace{-2.5ex}
\tikz[baseline]{
    \draw [line width=0.12ex]
        (0, 1.2ex) -- 
        (1.2ex, 1.2ex) --
        (1.2ex, 0.4ex) --
        (0.6ex, -0.25ex) --
        (0, 0.4ex) --
        cycle
}
\hspace{-1.1ex}}}\xspace \models \phi$).
\end{theorem}

\begin{proof}
  Assume $\lts$, $\phi$, and local shields $\localshield_i$ as in the assumptions.
  Observe that for $i = 1$, $\bigcap_{j < i} \phi_i = \states$, and that $\lts_{\ensuremath{\text{\protect 
{\color{white} $\nabla$} 
\hspace{-2.5ex}
\tikz[baseline]{
    \draw [line width=0.12ex]
        (0, 1.2ex) -- 
        (1.2ex, 1.2ex) --
        (1.2ex, 0.4ex) --
        (0.6ex, -0.25ex) --
        (0, 0.4ex) --
        cycle
}
\hspace{-1.1ex}}}\xspace^*[\states]} = \lts$.
  Then:
  \begin{align*}
    & \bigwedge_i  \left(\lts_{\ensuremath{\text{\protect 
{\color{white} $\nabla$} 
\hspace{-2.5ex}
\tikz[baseline]{
    \draw [line width=0.12ex]
        (0, 1.2ex) -- 
        (1.2ex, 1.2ex) --
        (1.2ex, 0.4ex) --
        (0.6ex, -0.25ex) --
        (0, 0.4ex) --
        cycle
}
\hspace{-1.1ex}}}\xspace^*[\bigcap_{j < i} \phi_j]} \right)^i \models_{\localshield_i} \cprj_i(\phi_i) \\
    \overset{\text{Lem.~\ref{lem:proj}}}{\implies} 
      & \bigwedge_i \lts_{\ensuremath{\text{\protect 
{\color{white} $\nabla$} 
\hspace{-2.5ex}
\tikz[baseline]{
    \draw [line width=0.12ex]
        (0, 1.2ex) -- 
        (1.2ex, 1.2ex) --
        (1.2ex, 0.4ex) --
        (0.6ex, -0.25ex) --
        (0, 0.4ex) --
        cycle
}
\hspace{-1.1ex}}}\xspace^*[\bigcap_{j < i} \phi_j]} \models_{\ext{\localshield_i}} \phi_i
    \overset{(*)}{\implies}
       \bigwedge_i \lts_{\comp_{j < i} \ext{\localshield_j}} \models_{\ext{\localshield_i}} \phi_i \\
    \overset{\text{Def.~\ref{def:models_shield}}}{\implies} 
      & \bigwedge_i \lts \models_{\comp_{j \le i} \ext{\localshield_j}} \phi_i
    \implies
       \lts \models_{\comp_i \, \ext{\localshield_i}} \phi
    \overset{\text{Def.~\ref{def:models_shield}}}{\implies}
       \lts_{\ensuremath{\text{\protect 
{\color{white} $\nabla$} 
\hspace{-2.5ex}
\tikz[baseline]{
    \draw [line width=0.12ex]
        (0, 1.2ex) -- 
        (1.2ex, 1.2ex) --
        (1.2ex, 0.4ex) --
        (0.6ex, -0.25ex) --
        (0, 0.4ex) --
        cycle
}
\hspace{-1.1ex}}}\xspace} \models \phi
  \end{align*}

  Step~$(*)$ holds because the composition $\comp_{j \leq i} \ext{\localshield_j}$ of the local shields up to index $i$ satisfy $\phi_i$ under the previous guarantees $\phi_j$, $j < i$.
  Thus, $\lts_{\comp_{j < i} \ext{\localshield_j}} \preceq \lts_{\ensuremath{\text{\protect 
{\color{white} $\nabla$} 
\hspace{-2.5ex}
\tikz[baseline]{
    \draw [line width=0.12ex]
        (0, 1.2ex) -- 
        (1.2ex, 1.2ex) --
        (1.2ex, 0.4ex) --
        (0.6ex, -0.25ex) --
        (0, 0.4ex) --
        cycle
}
\hspace{-1.1ex}}}\xspace^*[\bigcap_{j < i} \phi_j]}$, and the conclusion follows by applying Lemma~\ref{lem:preceq}.
\end{proof}

Finding the local safety properties~$\phi_i$ is an art, and we leave algorithmic synthesis of these properties to future work.
But we will show in our case studies that natural choices often exist, sometimes directly obtained from the (global) safety property.

\section{Cascading Learning}\label{sect:learning}

In the previous section, we have seen how to efficiently compute a distributed shield based on assume-guarantee reasoning.
In this section, we turn to the question how and under which condition we can efficiently learn multi-agent policies in a similar manner.

We start by defining the multi-agent learning objective.

\begin{definition}[$n$-agent cost function]
  Given an $n$-agent MDP $\mdp = (\states, \act, P)$ with projections~$\prj_i \colon \states \to O_i$, an \emph{$n$-agent cost function} $c = (c_1, \dots, c_n)$ consists of (local) cost functions $c_i \colon O_i \times \act_i \to \R$.
  The total immediate cost $c \colon \states \times \act \to \R$ is $c(s, a) = \sum_{i=1}^n c_i(\prj_i(s), a[i])$ for $s \in \states$ and $a \in \act$.
\end{definition}

An agent policy is obtained by projection, analogous to a local shield.
Next, we define the notion of instantiating an $n$-agent MDP with a policy, yielding an $(n-1)$-agent MDP.

\begin{definition}[Instantiating an agent]
  Given an $n$-agent MDP $\mdp = (\states, \act, P)$ and agent policy $\policy \colon O_i \times \act_i \to [0, 1]$, the \emph{instantiated MDP} is $\mdp_{\policy} = (\states, \act', P')$, where
  $\act' = \act_1 \times \dots \times \act_{i-1} \times \act_{i+1} \times \dots \times \act_n$
  and, for all $s, s' \in \states$ and $a' \in \act'$, $P'(s, a', s') = \sum_{a_i} \! \policy(\prj_i(s), a_i) \cdot P(s, (a'[1], \dots, a'[i-1], a_i, a'[i], \dots, a'[n-1]), s')$.
\end{definition}

We will need the concept of a projected, local run of an agent.

\begin{definition}[Local run]
  Given a run $\rho = s_0 a_0 s_1 a_1 \dots$ over an $n$-agent MDP $(\states, \act, P)$, the projection to agent~$i$ is the \emph{local run} $\prj_i(\rho) = \prj_i(s_0) \, a_0[i] \, \prj_i(s_1) \, a_1[i] \dots$
\end{definition}

Given a policy~$\policy \colon \states \times \act \to [0,1]$, the probability of a finite local run~$\prj_i(\rho)$ being an outcome of $\policy$ is the sum of the probabilities of outcomes of $\policy$ whose projection to $i$ is $\prj_i(\rho)$.

The probability of a run $\rho$ of length~$\ell$ being an outcome of policy~$\policy$ is $\prob(\rho \mid \policy) = \prod_{i=0} \policy(s_i, a_i) \cdot P(s_i, a_i, s_{i+1})$.
We say that agent~$i$ depends on agent~$j$ if agent~$j$'s action choice influences the probability for agent~$i$ to observe a (local) run.

\begin{definition}[Dependency]
  Given an $n$-agent MDP $(\states, \act, P)$, agent~$i$ \emph{depends} on agent~$j$ if
  there exists a local run~$\prj_i(\rho)$ of length~$\ell$ and $n$-agent policies~$\policy, \policy'$ that differ only in the $j$-th agent policy, i.e., $\policy = (\policy_1, \dots, \policy_n)$ and $\policy' = (\policy_1, \dots, \policy_{j-1}, \policy_j', \policy_{j+1}, \dots, \policy_n)$, such that the probability of observing~$\prj_i(\rho)$ under~$\pi$ and~$\pi'$ differ:
  \begin{align*}
    \sum\nolimits_{\rho' \colon \prj_i(\rho') = \prj_i(\rho)} \prob(\rho' \mid \policy) \ne \sum\nolimits_{\rho' \colon \prj_i(\rho') = \prj_i(\rho)} \prob(\rho' \mid \policy')
  \end{align*}
  where we sum over all runs $\rho'$ of length~$\ell$ with the same projection.
\end{definition}

In practice, we can typically perform an equivalent syntactic check.
Next, we show how to arrange dependencies in a graph.

\begin{definition}[Dependency graph]
  The \emph{dependency graph} of an $n$-agent MDP is a directed graph $(V, E)$ where $V = \{1, \dots, n\}$ and $E = \{(i, j) \mid i \text{ depends on } j\}$.
\end{definition}

As the main contribution of this section, Algorithm~\ref{algo:learn} shows an efficient multi-agent learning framework, which we call \emph{cascading learning}.
In order to apply the algorithm, we require an acyclic dependency graph (otherwise, an error is thrown in line~\ref{line:error}).
Then, we train the agents in the order suggested by the dependencies, which, as we will see, leads to an attractive property.

\begin{algorithm}[t]
  \caption{Cascading shielded learning of $n$-agent policies\hspace*{-4mm}}\label{algo:learn}

  \Input{Shielded $n$-agent MDP $\mdp_\ensuremath{\text{\protect 
{\color{white} $\nabla$} 
\hspace{-2.5ex}
\tikz[baseline]{
    \draw [line width=0.12ex]
        (0, 1.2ex) -- 
        (1.2ex, 1.2ex) --
        (1.2ex, 0.4ex) --
        (0.6ex, -0.25ex) --
        (0, 0.4ex) --
        cycle
}
\hspace{-1.1ex}}}\xspace$, \\ $n$-agent cost function $c = (c_1, \dots, c_n)$}
  \Output{$n$-agent policy $(\policy_1, \dots, \policy_n)$}

  \BlankLine

  Build dependency graph $G$ of $\mdp_\ensuremath{\text{\protect 
{\color{white} $\nabla$} 
\hspace{-2.5ex}
\tikz[baseline]{
    \draw [line width=0.12ex]
        (0, 1.2ex) -- 
        (1.2ex, 1.2ex) --
        (1.2ex, 0.4ex) --
        (0.6ex, -0.25ex) --
        (0, 0.4ex) --
        cycle
}
\hspace{-1.1ex}}}\xspace$\;

  Let $\mdp' := \mdp_\ensuremath{\text{\protect 
{\color{white} $\nabla$} 
\hspace{-2.5ex}
\tikz[baseline]{
    \draw [line width=0.12ex]
        (0, 1.2ex) -- 
        (1.2ex, 1.2ex) --
        (1.2ex, 0.4ex) --
        (0.6ex, -0.25ex) --
        (0, 0.4ex) --
        cycle
}
\hspace{-1.1ex}}}\xspace$\;

  \While{true}{
    \If{there is no node in $G$ with no outgoing edges}{
      error(``Cyclic dependencies are incompatible.'')\; \label{line:error}
    }
    Let $i$ be a node in $G$ with no outgoing edges\;

    Train agent policy~$\policy_i$ on the MDP $\sandbox(\mdp', i)$ wrt.\ cost function $c_i$\; \label{line:learn}

    Update $G$ by removing node $i$ and all incoming edges\;

    \lIf{$G$ is empty}{
      \Return{$(\policy_1, \dots, \policy_n)$}
    }

    Update $\mdp' := \mdp'_{\policy_i}$ \tcp*{(i.e., instantiated shielded MDP)}
  }
\end{algorithm}

To draw the connection to the distributed shield, the crucial insight is that we can again use it for assume-guarantee reasoning to prevent behaviors that may otherwise create a dependency.

The procedure $\sandbox(\mdp, i)$ in line~\ref{line:learn} takes an $n$-agent MDP~$\mdp$ and an agent index $i \in \{1, \dots, n\}$.
The purpose is to instantiate every agent except agent~$i$.
Since agent~$i$ does not depend on these agents, we arbitrary choose a uniform policy for the instantiation.

\smallskip

Next, we show an important property of Algorithm~\ref{algo:learn}: it trains policies in-distribution.

\begin{definition}[In-distribution]
  Given two $1$-agent MDPs $\mdp = (\states, \act, P)$ and $\mdp' = (\states, \act, P')$, an agent policy~$\policy$ is \emph{in-distribution} if the probability of any local run in $\mdp$ is the same as in $\mdp'$.
\end{definition}

Now we show that the distribution of observations an agent policy~$\policy_i$ makes during training in Algorithm~\ref{algo:learn} is identical with the distribution of observations made in~$\mdp^*$, the instantiation with \emph{all other} agent policies computed by Algorithm~\ref{algo:learn}.

\begin{theorem}
  Let $\mdp$ be an $n$-agent MDP with acyclic dependency graph.
  For every agent~$i$, the following holds.
  Let $\mdp^*$ be the $1$-agent MDP obtained by iteratively instantiating the original MDP $\mdp$ with policies $\policy_j$ for all $j \ne i$.
  The agent policy~$\policy_i$ trained with Algorithm~\ref{algo:learn} is in-distribution wrt.\ $\sandbox(\mdp', i)$ (from line~\ref{line:learn}) and $\mdp^*$.
\end{theorem}

\begin{proof}
  Fix a policy~$\policy_i$.
  If $\policy_i$ is the last trained policy, the statement clearly holds.
  Otherwise, let $\policy_j \ne \policy_i$ be a policy that has not been trained at the time when $\policy_i$ is trained.
  The algorithm asserts that $\policy_i$ has no dependency on $\policy_j$.
  Thus, training $\policy_i$ yields the same policy no matter how $\policy_j$ behaves.
\end{proof}

Note that, despite trained in-distribution, the policies are not globally optimal.
This is because each policy acts egoistically and optimizes its local cost, which may yield suboptimal global cost.

What we can show is that the agent policies $(\policy_1, \dots, \policy_n)$ are \emph{Pareto optimal}~\cite{marl-book}, i.e., they cannot all be strictly improved without raising the cost of at least one agent.
That is, there is no policy~$\policy_i$ that can be replaced by another policy~$\policy_i'$ without strictly increasing the expected local cost of at least one agent.
Indeed:

\begin{theorem}
  If the learning method in line~\ref{line:learn} of Algorithm~\ref{algo:learn} converged to the (local) optima, and these optima are unique, then the resulting policies are Pareto optimal.
\end{theorem}

\begin{proof}
  The proof is by induction.
  Assume \emph{wlog} that the policies are trained in the order 1 to~$n$.
  By assumption, $\pi_1$ is locally optimal and unique.
  Hence, replacing~$\pi_1$ by another policy would strictly increase its total cost.
  Now assume we have shown the claim for the first~$i-1$ agents.
  Algorithm~\ref{algo:learn} trained policy~$\policy_i$ wrt.\ the instantiation with the policies~$\pi_1, \dots, \pi_{i-1}$, and by assumption, $\policy_i$ is also locally optimal and unique.
  Thus, again, we cannot replace~$\policy_i$.
\end{proof}

\section{Evaluation}\label{sect:evaluation}

We consider two environments with discretized state spaces.%
\footnote{Available online at \url{https://github.com/AsgerHB/N-player-shield}.}
All experiments were repeated 10 times; solid lines in plots represent the mean cost of these 10 repetitions, while ribbons mark the minimum and maximum costs. Costs are evaluated as the mean of $1{,}000$ episodes.
We use the learning method implemented in \uppaalstratego~\cite{DBLP:conf/atva/JaegerJLLST19} because the implementation has a native interface for shields.
This method learns a policy by partition refinement of the state space.
With this learning method, only few episodes are needed for convergence.
We also compare to the (deep) MARL approach MAPPO~\cite{DBLP:conf/nips/YuVVGWBW22} later.

\subsection{Car Platoon with Adaptive Cruise Controls}

Recall the car platoon model from Section~\ref{sect:platoon}.
The front car follows a random distribution depending on $v_n$ (described in\refappendix{Appendix~\ref{sect:frontcarspeed}}).

The individual cost of an agent is the sum of the observed distances to the car immediately in front of it, during a 100-second episode
(i.e., keeping a smaller distance to the car in front is better).

The decision period causes delayed reaction time, and so the minimum safe distance to the car in front depends on the velocity of both cars.
An agent must learn to drive up to this distance,
and then maintain it
by predicting the acceleration of the car in front.

For this model, all agents share analogous observations~$O_i$ and safety properties~$\phi_i$.
Hence, instead of computing $n-1$ local shields individually, it is sufficient to compute only one local shield and reuse it across all agents (by simply adapting the variables).

\subsubsection{Relative scalability of centralized and distributed shielding}

We compare the synthesis of distributed and (non-distributed) classical shields.
We call the latter \emph{centralized} shields, as they reason about the global state. Hence, they may permit more behavior and potentially lead to better policies, as the agents can coordinate to take jointly safe actions. 
Beside this (often unrealistic) coordination assumption, a centralized shield suffers from scalability issues.
While the size of a single agent's observation space is modest, the global state space is often too large for computing a shield.

We interrupted the synthesis of a centralized shield with $n=3$ cars (i.e., 2 agents) and a full state space after 12 hours, at which point the computation showed less than 3\% progress. 
In order to obtain a centralized shield, we reduced the maximum safe distance from~200 to just~50, shrinking the state space significantly. Synthesizing a centralized shield took 78 minutes for this property, compared to just 3 seconds for a corresponding distributed shield.

Because of the exponential complexity to synthesize a centralized shield, we will only consider distributed shields in the following. 
Synthesizing a shield for a single agent covering the full safety property ($0 < d_i < 200$) took 6.5 seconds, which we will apply to a platoon of 10 cars, well out of reach of a centralized shield.

\subsubsection{Comparing centralized, cascading and MAPPO learning}

\begin{figure}[t]
  \centering
  \includesvg[width=\linewidth]{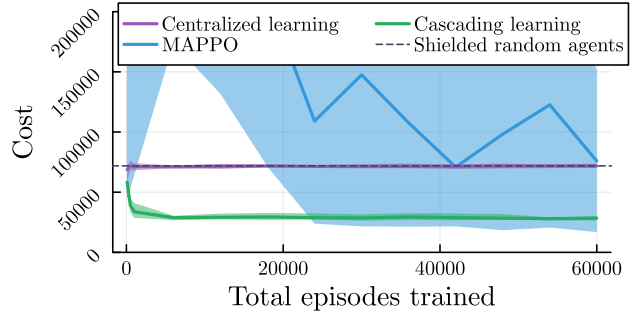}
  \caption{Graph of learning outcomes. Comparison of different learning methods on the 10-car platoon. The centralized and the MAPPO policy were trained for the total episodes indicated, while these episodes were split evenly between each agent in the cascading case.}
  \Description{Graph showing cost as a function of training episodes.
  The (shielded) centralized learning policy has almost no difference between minimum and maximum cost, nearly overlapping with the dashed horizontal line which shows the performance of shielded random agents. The MAPPO policy has a maximum cost that exceeds the bounds of the graph in all cases, while the mean average cost improves with more training and sometimes matches the shielded random agents and the (shielded) centralized learning. The minimum MAPPO outcome achieves a cost a bit lower than the (shielded) cascading learning, shortly after 20,000 training episodes. As for the (shielded) cascading learning, in just a few thousand episodes of training, it reaches a plateau with very little difference between minimum and maximum training outcome.}
  \label{fig:cclearning}
\end{figure}

Given a distributed shield, we consider the learning outcomes for a platoon of 10 cars (9 agents), using the learning method of \uppaalstratego.
We train both a shielded \emph{centralized} policy, which picks a joint action for all cars, and individual shielded policies using cascading learning (Algorithm~\ref{algo:learn}).
As expected from shielded policies, no safety violations were observed while evaluating them.

In the results shown in Figure~\ref{fig:cclearning}, the centralized policy does not improve with more training. 
While it could theoretically outperform distributed policies through communication, the high dimensionality of the state and action space likely prevents that.
It only marginally improves over the random baseline, which has an average cost of~$71{,}871$.
On the other hand, cascading learning quickly converges to a much better cost as low as $26{,}435$.

To examine how cascading learning under a distributed shield compares to traditional MARL techniques,
we implemented the platoon environment in the benchmark suite BenchMARL~\cite{DBLP:journals/corr/abs-2312-01472} and trained an unshielded policy with MAPPO~\cite{DBLP:conf/nips/YuVVGWBW22}, a state-of-the-art MARL algorithm based on PPO~\cite{DBLP:journals/corr/SchulmanWDRK17}, using default hyperparameters.
To encourage safe behavior, we added a penalty of $1{,}600$ to the cost function for every step upon reaching an unsafe state. 
(This value was obtained by starting from $100$ and doubling it until safety started degrading again.)
Recall that shielded agents are safe.

We include the training outcomes for MAPPO in Figure~\ref{fig:cclearning}.
Due primarily to the penalty of safety violations, the agents often have a cost greater than $100{,}000$, even at the end of training.
However, the best MAPPO policy achieved a cost of just $16{,}854$, better than the cascading learning method.
We inspected that policy and found that the cars drive very closely, accepting the risk of a crash.
Overall, there is a large variance of the MAPPO policies in different runs, whereas cascading learning converges to very similar policies, and does so much faster.
This is likely because of the smaller space in which the policies are learned, due to the distributed shield.
Thus, cascading learning is more effective.

\begin{figure}[t]
  \centering
  \includesvg[width=0.8\linewidth]{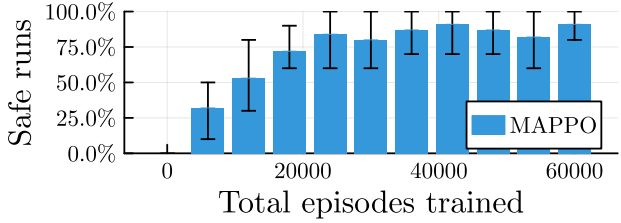}
  \caption{Percentage of safe runs with the MAPPO policy in the 10-car platoon. Blue bars show the mean of 10 repetitions, while black intervals give min and max values.}
  \Description{Bar chart. After 20,000 training episodes, the best-case training outcome shows 100 percent safe runs observed, while the worst case is still unsafe. This trend continues up to 60,000 episodes, where the worst case training outcome still does not result in 100 percent safe traces. The mean safety does not improve noticeably between somewhere around 40,000 and up until the end.}
  \label{fig:ccmappopercentagesafe}
\end{figure}

Since the MAPPO policy is not safe by construction, Figure~\ref{fig:ccmappopercentagesafe} shows the percentage of safe episodes, out of $1{,}000$ episodes.
The agents tend to be safer with more training, but there is no inherent guarantee of safety, and a significant amount of violations remain.

\subsection{Chemical Production Plant}

In the second case study, we demonstrate that distributed shielding applies to complex dependencies where agents influence multiple other agents asymmetrically.
We consider a network of inter-connected chemical production units, each with an internal storage.

\begin{wrapfigure}{r}{.33\linewidth}
  \vspace*{-2mm}
  \centering
  \begin{tikzpicture}[
    node distance=0.3cm and 0.3cm,
    auto,
    blank/.style={
        minimum height=0.4cm
    },
    unit/.style={
        rectangle,
        rounded corners,
        draw=black,
        align=center,
        minimum height=0.4cm
    },
    consumer/.style={
        circle,
        draw=black,
        align=center,
        minimum height=0.1cm
    },
    every path/.style={
        -Latex
    }
]

\node [unit] (u1) {1};
\node [blank, right=of u1] (b1) {};
\node [unit, right=of b1] (u2) {2};
\node [blank, right=of u2] (b2) {};
\node [unit, right=of b2] (u3) {3};

\node [blank, below=of u1] (b3) {};
\node [unit, right=of b3] (u4) {4};
\node [blank, right=of u4] (b4) {};
\node [unit, right=of b4] (u5) {5};
\node [blank, right=of u5] (b5) {};

\node [unit, below=of b3] (u6) {6};
\node [blank, right=of u6] (b6) {};
\node [unit, right=of b6] (u7) {7};
\node [blank, right=of u7] (b7) {};
\node [unit, right=of b7] (u8) {8};

\node [blank, below=of u6] (b8) {};
\node [unit, right=of b8] (u9) {9};
\node [blank, right=of u9] (b9) {};
\node [unit, right=of b9] (u10) {10};
\node [blank, right=of u10] (b10) {};

\node [consumer, below=of u9] (c1) {\small A};
\node [consumer, below=of u10] (c2) {\small B};

\path (u1) edge (u4)
    (u2) edge (u4)
    (u2) edge (u5)
    (u3) edge (u5)
    (u4) edge (u6)
    (u4) edge (u7)
    (u5) edge (u7)
    (u5) edge (u8)
    (u6) edge (u9)
    (u7) edge (u9)
    (u7) edge (u10)
    (u8) edge (u10)

    ($ (b1) + (-0.1, 0.14) $) edge ($ (u1) + (0.18, 0.14) $)
    ($ (b1) + (-0.1, -0.0) $) edge ($ (u1) + (0.18, -0.0) $)
    ($ (b1) + (-0.1, -0.14) $) edge ($ (u1) + (0.18, -0.14) $)
    ($ (b2) + (-0.1, 0.14) $) edge ($ (u2) + (0.18, 0.14) $)
    ($ (b2) + (-0.1, -0.0) $) edge ($ (u2) + (0.18, -0.0) $)
    ($ (b2) + (-0.1, -0.14) $) edge ($ (u2) + (0.18, -0.14) $)
    ($ (b2) + (0.1, 0.14) $) edge ($ (u3) + (-0.18, 0.14) $)
    ($ (b2) + (0.1, -0.0) $) edge ($ (u3) + (-0.18, -0.0) $)
    ($ (b2) + (0.1, -0.14) $) edge ($ (u3) + (-0.18, -0.14) $)
    (b4) edge (u4)
    (b5) edge (u5)
    ($ (b6) + (-0.1, 0.08) $) edge ($ (u6) + (0.18, 0.08) $)
    ($ (b6) + (-0.1, -0.08) $) edge ($ (u6) + (0.18, -0.08) $)
    (b7) edge (u7)
    ($ (b7) + (+0.1, 0.08) $) edge ($ (u8) + (-0.18, 0.08) $)
    ($ (b7) + (+0.1, -0.08) $) edge ($ (u8) + (-0.18, -0.08) $)
    (b9) edge (u9)
    (b10) edge (u10)

    (u9) edge ($ (c1) + (0.1, 0.25)$)
    (u9) edge ($ (c1) + (-0.1, 0.25)$)
    (u10) edge ($ (c2) + (0.1, 0.25)$)
    (u10) edge ($ (c2) + (-0.1, 0.25)$)
    ;

\end{tikzpicture}
  \caption{Layout of plant network.}
  \Description{Graph of inter-connected nodes. The nodes are numbered 1 through 10, plus two nodes labeled A and B. The numbered nodes all have three incoming arrows, and one or two outgoing arrows. All nodes are aligned in layers. Nodes 1, 2, and 3 appear in the top row. Each have three incoming arrows with no source nodes. Node 1 has an arrow from itself down to node 4. Similarly, node 2 has arrows to nodes 4 and 5. Node 4 has an arrow to node 5. The next row consists of node 4 and 5. Each have 1 additional incoming arrow with no source node. Node 4 has arrows to node 6 and 7. Node 5 has arrows to 7 and 8. The next row consists of nodes 6, 7, and 8. Node 6 and 8 have two additional incoming arrows with no source node, while note 7 only has one. Node 6 has an arrow to node 9. Node 7 has arrows to node 9 and 10. Node 8 has one arrow to node 10. The next row consists of nodes 9 and 10. Each have 1 incoming arrow with no source node. Node 9 has two outgoing arrows that both point to node A. Likewise, node 10 has two outgoing arrows that each point to node B. The last row consists of nodes A and B.}
  \label{fig:cp_layout}
\end{wrapfigure}
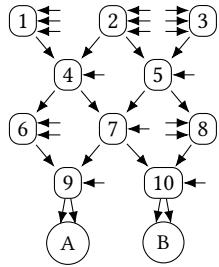
Figure~\ref{fig:cp_layout} shows the graph structure of the network.
Numbered nodes (1 to~10) denote controlled production units, while letter-labeled nodes (A, B) denote uncontrolled consumers with periodically varying demand. 
Arrows from source to target nodes denote potential flow at no incurred cost.
Arrows without a source node denote potential flow from external providers, at a cost that individually and periodically varies.
Consumption patterns \ifshowappendix(Figure~\ref{fig:consumerdemand}) \fi and examples of the cost patterns \ifshowappendix(Figure~\ref{fig:providercost}) \fi are shown in\refappendix{Appendix~\ref{sect:plant}}.
The flow rate in all arrows follows a uniform random distribution in the range $\interval{2.15}{3.15}$ $\ell$/s.

Each agent~$i$ is associated with a production unit (1 to~10), with internal storage volume~$v_i$.
Beside a global periodic timer, each agent can only observe its own volume.
At each decision period of 0.5 seconds, an agent can open or close each of the three input flows (i.e., there are $|\act_i| = 9$ actions per agent and hence $|\act| = 9^{10}$ global actions), but cannot prevent flow from outgoing connections.

The individual cost of an agent is incurred by buying from external providers. 
Agents must learn to take free material from other units, except for agents~1 to~3, which instead must learn to buy from their external providers periodically when the cost is low.

Units must not exceed their storage capacity, and units~9 to~10 must also not run empty to ensure the consumers' demand is met. That is, the safety property is $\phi = \{s \mid \bigwedge_i v_i < 50 \land 0 < v_9 \land 0 < v_{10}\}$.

\subsubsection{Shielding.}
The property~$0 < v_9$ cannot be enforced by a local shield for agent~$9$ without additional assumptions that the other agents do not run out.
This is because the (single) external provider is not enough to meet the potential (dual) demand of consumer~$A$.
This yields the local safety properties $\phi_i = \{s \mid 0 < v_i < 50\}$.
Here, agents~1 to~3 do not make assumptions, while agents~4 and~5 depend on agents~1 to~3 not running out, etc.
For this model, we do not use the same shield for all agents, since they differ in the number of outgoing flows (either 1 or~2).
Still, it is sufficient to compute two types of shields, one for each variant, and adapt them to analogous agents.
Computing a centralized shield would again be infeasible, while computing the distributed shield took less than 1 second.

\subsubsection{Comparing centralized and cascading learning}

\begin{figure}[t]
  \centering
  \includesvg[width=\linewidth]{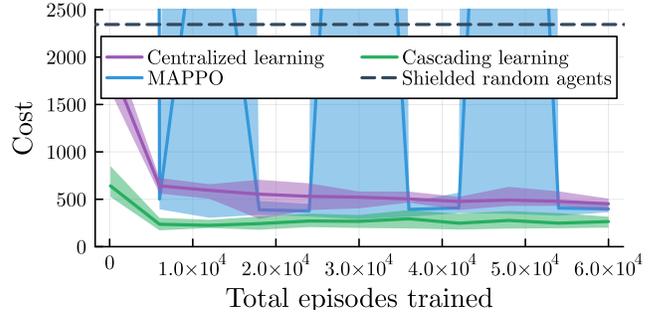}
  \caption{Comparison of different learning methods on the chemical production plant. The centralized policy was trained for the total episodes indicated, while these episodes were split evenly between each agent in the cascading case. }
  \Description{Graph showing cost as a function of training episodes. A dashed horizontal line shows the performance of shielded random agents as a baseline. Both the (shielded) centralized learning policy and the (shielded) cascading learning policy very quickly show improvement, before reaching a plateau at 10,000 episodes. The cascading learning policy has better cost than the centralized learning policy. Both policies perform much better than the shielded random agents. The MAPPO policy has 3 distinct spikes, where the maximum and mean costs rise far above that of the shielded random agents. These spikes appear at 10,000, 30,000 and 50,000 training episodes. In these 3 spikes, the minimum cost is unaffected. Beside the 3 spikes, the policy achieves a higher cost than the shielded cascading learning, but better than the shielded centralized learning.}
  \label{fig:cplearning}
\end{figure}

Thanks to the guarantees given by the  distributed shields, agents~9 to~10 are only affected by the behavior of the consumers, agents~6 to~8 only depend on agents~9 to~10, etc. Thus, the agent training order is 10, 9, 8 \dots

We compare the results of shielded cascading learning, shielded centralized learning, MAPPO, and shielded random agents in Figure~\ref{fig:cplearning}.
Centralized learning achieved a cost of~$292$.
The lowest cost overall, $172$, was achieved by cascading learning.
We compare this to the (unshielded) MAPPO agents, whose lowest cost was~$291$.
More background information is given in\refappendix{Appendix~\ref{sect:cpmapposfaety}}.

\section{Conclusion}\label{sect:conclusion}

In this paper, we presented distributed shielding as a scalable MA approach, which we made practically applicable by integrating assume-guarantee reasoning.
We also presented cascading shielded learning, which, when applicable, is a scalable MARL approach.
We demonstrated that distributed shield synthesis is highly scalable and that coming up with useful guarantees is reasonably simple.

While we focused on demonstrating the feasibility in this work by providing the guarantees manually, a natural future direction is to learn them.
As discussed, this is much simpler in the classical setting~\cite{DBLP:reference/mc/GiannakopoulouNP18} because the agents/components are fixed.
We believe that in our setting where both the guarantees and the agents are not given, a trial-and-error approach (e.g., a genetic algorithm) is a fruitful direction to explore.
Another relevant future direction is to generalize our approach to continuous systems~\cite{DBLP:conf/aisola/BrorholtJLLS23}.

\begin{acks}
This research was partly supported by the Independent Research Fund Denmark under reference number 10.46540/3120-00041B, DIREC - Digital Research Centre Denmark under reference number 9142-0001B, and the Villum Investigator Grant S4OS under reference number 37819.
\end{acks}

\bibliographystyle{ACM-Reference-Format}
\balance
\bibliography{bibliography}


\begin{thebibliography}{38}


\ifx \showCODEN    \undefined \def \showCODEN     #1{\unskip}     \fi
\ifx \showDOI      \undefined \def \showDOI       #1{#1}\fi
\ifx \showISBNx    \undefined \def \showISBNx     #1{\unskip}     \fi
\ifx \showISBNxiii \undefined \def \showISBNxiii  #1{\unskip}     \fi
\ifx \showISSN     \undefined \def \showISSN      #1{\unskip}     \fi
\ifx \showLCCN     \undefined \def \showLCCN      #1{\unskip}     \fi
\ifx \shownote     \undefined \def \shownote      #1{#1}          \fi
\ifx \showarticletitle \undefined \def \showarticletitle #1{#1}   \fi
\ifx \showURL      \undefined \def \showURL       {\relax}        \fi
\providecommand\bibfield[2]{#2}
\providecommand\bibinfo[2]{#2}
\providecommand\natexlab[1]{#1}
\providecommand\showeprint[2][]{arXiv:#2}

\bibitem[\protect\citeauthoryear{Abate, Gutierrez, Hammond, Harrenstein,
  Kwiatkowska, Najib, Perelli, Steeples, and Wooldridge}{Abate
  et~al\mbox{.}}{2021}]%
        {DBLP:journals/apin/AbateGHHKNPSW21}
\bibfield{author}{\bibinfo{person}{Alessandro Abate}, \bibinfo{person}{Julian
  Gutierrez}, \bibinfo{person}{Lewis Hammond}, \bibinfo{person}{Paul
  Harrenstein}, \bibinfo{person}{Marta Kwiatkowska}, \bibinfo{person}{Muhammad
  Najib}, \bibinfo{person}{Giuseppe Perelli}, \bibinfo{person}{Thomas
  Steeples}, {and} \bibinfo{person}{Michael~J. Wooldridge}.}
  \bibinfo{year}{2021}\natexlab{}.
\newblock \showarticletitle{Rational verification: game-theoretic verification
  of multi-agent systems}.
\newblock \bibinfo{journal}{\emph{Appl. Intell.}} \bibinfo{volume}{51},
  \bibinfo{number}{9} (\bibinfo{year}{2021}), \bibinfo{pages}{6569--6584}.
\newblock
\urldef\tempurl%
\url{https://doi.org/10.1007/S10489-021-02658-Y}
\showDOI{\tempurl}


\bibitem[\protect\citeauthoryear{Albrecht, Christianos, and Sch\"afer}{Albrecht
  et~al\mbox{.}}{2024}]%
        {marl-book}
\bibfield{author}{\bibinfo{person}{Stefano~V. Albrecht},
  \bibinfo{person}{Filippos Christianos}, {and} \bibinfo{person}{Lukas
  Sch\"afer}.} \bibinfo{year}{2024}\natexlab{}.
\newblock \bibinfo{booktitle}{\emph{Multi-Agent Reinforcement Learning:
  Foundations and Modern Approaches}}.
\newblock \bibinfo{publisher}{MIT Press}.
\newblock
\urldef\tempurl%
\url{https://www.marl-book.com}
\showURL{%
\tempurl}


\bibitem[\protect\citeauthoryear{Alshiekh, Bloem, Ehlers, K{\"{o}}nighofer,
  Niekum, and Topcu}{Alshiekh et~al\mbox{.}}{2018}]%
        {DBLP:conf/aaai/AlshiekhBEKNT18}
\bibfield{author}{\bibinfo{person}{Mohammed Alshiekh},
  \bibinfo{person}{Roderick Bloem}, \bibinfo{person}{R{\"{u}}diger Ehlers},
  \bibinfo{person}{Bettina K{\"{o}}nighofer}, \bibinfo{person}{Scott Niekum},
  {and} \bibinfo{person}{Ufuk Topcu}.} \bibinfo{year}{2018}\natexlab{}.
\newblock \showarticletitle{Safe Reinforcement Learning via Shielding}. In
  \bibinfo{booktitle}{\emph{{AAAI}}}. \bibinfo{publisher}{{AAAI} Press},
  \bibinfo{pages}{2669--2678}.
\newblock
\urldef\tempurl%
\url{https://doi.org/10.1609/AAAI.V32I1.11797}
\showDOI{\tempurl}


\bibitem[\protect\citeauthoryear{Bellemare, Candido, Castro, Gong, Machado,
  Moitra, Ponda, and Wang}{Bellemare et~al\mbox{.}}{2020}]%
        {DBLP:journals/nature/BellemareCCGMMP20}
\bibfield{author}{\bibinfo{person}{Marc~G. Bellemare},
  \bibinfo{person}{Salvatore Candido}, \bibinfo{person}{Pablo~Samuel Castro},
  \bibinfo{person}{Jun Gong}, \bibinfo{person}{Marlos~C. Machado},
  \bibinfo{person}{Subhodeep Moitra}, \bibinfo{person}{Sameera~S. Ponda}, {and}
  \bibinfo{person}{Ziyu Wang}.} \bibinfo{year}{2020}\natexlab{}.
\newblock \showarticletitle{Autonomous navigation of stratospheric balloons
  using reinforcement learning}.
\newblock \bibinfo{journal}{\emph{Nat.}} \bibinfo{volume}{588},
  \bibinfo{number}{7836} (\bibinfo{year}{2020}), \bibinfo{pages}{77--82}.
\newblock
\urldef\tempurl%
\url{https://doi.org/10.1038/S41586-020-2939-8}
\showDOI{\tempurl}


\bibitem[\protect\citeauthoryear{Benveniste, Caillaud, Nickovic, Passerone,
  Raclet, Reinkemeier, Sangiovanni{-}Vincentelli, Damm, Henzinger, and
  Larsen}{Benveniste et~al\mbox{.}}{2018}]%
        {DBLP:journals/fteda/BenvenisteCNPRR18}
\bibfield{author}{\bibinfo{person}{Albert Benveniste},
  \bibinfo{person}{Beno{\^{\i}}t Caillaud}, \bibinfo{person}{Dejan Nickovic},
  \bibinfo{person}{Roberto Passerone}, \bibinfo{person}{Jean{-}Baptiste
  Raclet}, \bibinfo{person}{Philipp Reinkemeier}, \bibinfo{person}{Alberto~L.
  Sangiovanni{-}Vincentelli}, \bibinfo{person}{Werner Damm},
  \bibinfo{person}{Thomas~A. Henzinger}, {and} \bibinfo{person}{Kim~G.
  Larsen}.} \bibinfo{year}{2018}\natexlab{}.
\newblock \showarticletitle{Contracts for System Design}.
\newblock \bibinfo{journal}{\emph{Found. Trends Electron. Des. Autom.}}
  \bibinfo{volume}{12}, \bibinfo{number}{2-3} (\bibinfo{year}{2018}),
  \bibinfo{pages}{124--400}.
\newblock
\urldef\tempurl%
\url{https://doi.org/10.1561/1000000053}
\showDOI{\tempurl}


\bibitem[\protect\citeauthoryear{Bettini, Prorok, and Moens}{Bettini
  et~al\mbox{.}}{2023}]%
        {DBLP:journals/corr/abs-2312-01472}
\bibfield{author}{\bibinfo{person}{Matteo Bettini}, \bibinfo{person}{Amanda
  Prorok}, {and} \bibinfo{person}{Vincent Moens}.}
  \bibinfo{year}{2023}\natexlab{}.
\newblock \showarticletitle{{BenchMARL}: Benchmarking Multi-Agent Reinforcement
  Learning}.
\newblock \bibinfo{journal}{\emph{CoRR}} (\bibinfo{year}{2023}).
\newblock
\showeprint[arXiv]{2312.01472}


\bibitem[\protect\citeauthoryear{Bharadwaj, Bloem, Dimitrova, K{\"{o}}nighofer,
  and Topcu}{Bharadwaj et~al\mbox{.}}{2019}]%
        {DBLP:conf/amcc/BharadwajBDKT19}
\bibfield{author}{\bibinfo{person}{Suda Bharadwaj}, \bibinfo{person}{Roderick
  Bloem}, \bibinfo{person}{Rayna Dimitrova}, \bibinfo{person}{Bettina
  K{\"{o}}nighofer}, {and} \bibinfo{person}{Ufuk Topcu}.}
  \bibinfo{year}{2019}\natexlab{}.
\newblock \showarticletitle{Synthesis of Minimum-Cost Shields for Multi-agent
  Systems}. In \bibinfo{booktitle}{\emph{{ACC}}}. \bibinfo{publisher}{{IEEE}},
  \bibinfo{pages}{1048--1055}.
\newblock
\urldef\tempurl%
\url{https://doi.org/10.23919/ACC.2019.8815233}
\showDOI{\tempurl}


\bibitem[\protect\citeauthoryear{Bloem, Chatterjee, and Jobstmann}{Bloem
  et~al\mbox{.}}{2018}]%
        {DBLP:reference/mc/BloemCJ18}
\bibfield{author}{\bibinfo{person}{Roderick Bloem}, \bibinfo{person}{Krishnendu
  Chatterjee}, {and} \bibinfo{person}{Barbara Jobstmann}.}
  \bibinfo{year}{2018}\natexlab{}.
\newblock \showarticletitle{Graph Games and Reactive Synthesis}.
\newblock In \bibinfo{booktitle}{\emph{Handbook of Model Checking}}.
  \bibinfo{publisher}{Springer}, \bibinfo{pages}{921--962}.
\newblock
\urldef\tempurl%
\url{https://doi.org/10.1007/978-3-319-10575-8\_27}
\showDOI{\tempurl}


\bibitem[\protect\citeauthoryear{Brorholt, H{\o}eg{-}Petersen, Larsen, and
  Schilling}{Brorholt et~al\mbox{.}}{2024}]%
        {DBLP:conf/aisola/BrorholtHLS24}
\bibfield{author}{\bibinfo{person}{Asger~Horn Brorholt},
  \bibinfo{person}{Andreas~Holck H{\o}eg{-}Petersen},
  \bibinfo{person}{Kim~Guldstrand Larsen}, {and} \bibinfo{person}{Christian
  Schilling}.} \bibinfo{year}{2024}\natexlab{}.
\newblock \showarticletitle{Efficient Shield Synthesis via State-Space
  Transformation}. In \bibinfo{booktitle}{\emph{{AISoLA}}}
  \emph{(\bibinfo{series}{LNCS}, Vol.~\bibinfo{volume}{15217})}.
  \bibinfo{publisher}{Springer}, \bibinfo{pages}{206--224}.
\newblock
\urldef\tempurl%
\url{https://doi.org/10.1007/978-3-031-75434-0\_14}
\showDOI{\tempurl}


\bibitem[\protect\citeauthoryear{Brorholt, Jensen, Larsen, Lorber, and
  Schilling}{Brorholt et~al\mbox{.}}{2023}]%
        {DBLP:conf/aisola/BrorholtJLLS23}
\bibfield{author}{\bibinfo{person}{Asger~Horn Brorholt},
  \bibinfo{person}{Peter~Gj{\o}l Jensen}, \bibinfo{person}{Kim~Guldstrand
  Larsen}, \bibinfo{person}{Florian Lorber}, {and} \bibinfo{person}{Christian
  Schilling}.} \bibinfo{year}{2023}\natexlab{}.
\newblock \showarticletitle{Shielded Reinforcement Learning for Hybrid
  Systems}. In \bibinfo{booktitle}{\emph{{AISoLA}}}
  \emph{(\bibinfo{series}{LNCS}, Vol.~\bibinfo{volume}{14380})}.
  \bibinfo{publisher}{Springer}, \bibinfo{pages}{33--54}.
\newblock
\urldef\tempurl%
\url{https://doi.org/10.1007/978-3-031-46002-9\_3}
\showDOI{\tempurl}


\bibitem[\protect\citeauthoryear{Carr, Jansen, Junges, and Topcu}{Carr
  et~al\mbox{.}}{2023}]%
        {DBLP:conf/aaai/Carr0JT23}
\bibfield{author}{\bibinfo{person}{Steven Carr}, \bibinfo{person}{Nils Jansen},
  \bibinfo{person}{Sebastian Junges}, {and} \bibinfo{person}{Ufuk Topcu}.}
  \bibinfo{year}{2023}\natexlab{}.
\newblock \showarticletitle{Safe Reinforcement Learning via Shielding under
  Partial Observability}. In \bibinfo{booktitle}{\emph{{AAAI}}}.
  \bibinfo{publisher}{{AAAI} Press}, \bibinfo{pages}{14748--14756}.
\newblock
\urldef\tempurl%
\url{https://doi.org/10.1609/AAAI.V37I12.26723}
\showDOI{\tempurl}


\bibitem[\protect\citeauthoryear{Chatterjee, Chmelik, and Tracol}{Chatterjee
  et~al\mbox{.}}{2016}]%
        {DBLP:journals/jcss/ChatterjeeCT16}
\bibfield{author}{\bibinfo{person}{Krishnendu Chatterjee},
  \bibinfo{person}{Martin Chmelik}, {and} \bibinfo{person}{Mathieu Tracol}.}
  \bibinfo{year}{2016}\natexlab{}.
\newblock \showarticletitle{What is decidable about partially observable Markov
  decision processes with {\(\omega\)}-regular objectives}.
\newblock \bibinfo{journal}{\emph{J. Comput. Syst. Sci.}} \bibinfo{volume}{82},
  \bibinfo{number}{5} (\bibinfo{year}{2016}), \bibinfo{pages}{878--911}.
\newblock
\urldef\tempurl%
\url{https://doi.org/10.1016/J.JCSS.2016.02.009}
\showDOI{\tempurl}


\bibitem[\protect\citeauthoryear{David, Jensen, Larsen, Legay, Lime,
  S{\o}rensen, and Taankvist}{David et~al\mbox{.}}{2014}]%
        {DBLP:conf/atva/DavidJLLLST14}
\bibfield{author}{\bibinfo{person}{Alexandre David},
  \bibinfo{person}{Peter~Gj{\o}l Jensen}, \bibinfo{person}{Kim~Guldstrand
  Larsen}, \bibinfo{person}{Axel Legay}, \bibinfo{person}{Didier Lime},
  \bibinfo{person}{Mathias~Grund S{\o}rensen}, {and}
  \bibinfo{person}{Jakob~Haahr Taankvist}.} \bibinfo{year}{2014}\natexlab{}.
\newblock \showarticletitle{On Time with Minimal Expected Cost!}. In
  \bibinfo{booktitle}{\emph{{ATVA}}} \emph{(\bibinfo{series}{LNCS},
  Vol.~\bibinfo{volume}{8837})}. \bibinfo{publisher}{Springer},
  \bibinfo{pages}{129--145}.
\newblock
\urldef\tempurl%
\url{https://doi.org/10.1007/978-3-319-11936-6\_10}
\showDOI{\tempurl}


\bibitem[\protect\citeauthoryear{Elsayed{-}Aly, Bharadwaj, Amato, Ehlers,
  Topcu, and Feng}{Elsayed{-}Aly et~al\mbox{.}}{2021}]%
        {DBLP:conf/atal/Elsayed-AlyBAET21}
\bibfield{author}{\bibinfo{person}{Ingy Elsayed{-}Aly}, \bibinfo{person}{Suda
  Bharadwaj}, \bibinfo{person}{Christopher Amato},
  \bibinfo{person}{R{\"{u}}diger Ehlers}, \bibinfo{person}{Ufuk Topcu}, {and}
  \bibinfo{person}{Lu Feng}.} \bibinfo{year}{2021}\natexlab{}.
\newblock \showarticletitle{Safe Multi-Agent Reinforcement Learning via
  Shielding}. In \bibinfo{booktitle}{\emph{{AAMAS}}}.
  \bibinfo{publisher}{{ACM}}, \bibinfo{pages}{483--491}.
\newblock
\urldef\tempurl%
\url{https://doi.org/10.5555/3463952.3464013}
\showDOI{\tempurl}


\bibitem[\protect\citeauthoryear{Giannakopoulou, Namjoshi, and
  Pasareanu}{Giannakopoulou et~al\mbox{.}}{2018}]%
        {DBLP:reference/mc/GiannakopoulouNP18}
\bibfield{author}{\bibinfo{person}{Dimitra Giannakopoulou},
  \bibinfo{person}{Kedar~S. Namjoshi}, {and} \bibinfo{person}{Corina~S.
  Pasareanu}.} \bibinfo{year}{2018}\natexlab{}.
\newblock \showarticletitle{Compositional Reasoning}.
\newblock In \bibinfo{booktitle}{\emph{Handbook of Model Checking}}.
  \bibinfo{publisher}{Springer}, \bibinfo{pages}{345--383}.
\newblock
\urldef\tempurl%
\url{https://doi.org/10.1007/978-3-319-10575-8\_12}
\showDOI{\tempurl}


\bibitem[\protect\citeauthoryear{Jaeger, Jensen, Larsen, Legay, Sedwards, and
  Taankvist}{Jaeger et~al\mbox{.}}{2019}]%
        {DBLP:conf/atva/JaegerJLLST19}
\bibfield{author}{\bibinfo{person}{Manfred Jaeger},
  \bibinfo{person}{Peter~Gj{\o}l Jensen}, \bibinfo{person}{Kim~Guldstrand
  Larsen}, \bibinfo{person}{Axel Legay}, \bibinfo{person}{Sean Sedwards}, {and}
  \bibinfo{person}{Jakob~Haahr Taankvist}.} \bibinfo{year}{2019}\natexlab{}.
\newblock \showarticletitle{Teaching {S}tratego to Play Ball: Optimal Synthesis
  for Continuous Space {MDP}s}. In \bibinfo{booktitle}{\emph{{ATVA}}}
  \emph{(\bibinfo{series}{LNCS}, Vol.~\bibinfo{volume}{11781})}.
  \bibinfo{publisher}{Springer}, \bibinfo{pages}{81--97}.
\newblock
\urldef\tempurl%
\url{https://doi.org/10.1007/978-3-030-31784-3\_5}
\showDOI{\tempurl}


\bibitem[\protect\citeauthoryear{Jansen, K{\"{o}}nighofer, Junges, Serban, and
  Bloem}{Jansen et~al\mbox{.}}{2020}]%
        {DBLP:conf/concur/0001KJSB20}
\bibfield{author}{\bibinfo{person}{Nils Jansen}, \bibinfo{person}{Bettina
  K{\"{o}}nighofer}, \bibinfo{person}{Sebastian Junges}, \bibinfo{person}{Alex
  Serban}, {and} \bibinfo{person}{Roderick Bloem}.}
  \bibinfo{year}{2020}\natexlab{}.
\newblock \showarticletitle{Safe Reinforcement Learning Using Probabilistic
  Shields}. In \bibinfo{booktitle}{\emph{{CONCUR}}},
  Vol.~\bibinfo{volume}{171}. \bibinfo{pages}{3:1--3:16}.
\newblock
\urldef\tempurl%
\url{https://doi.org/10.4230/LIPICS.CONCUR.2020.3}
\showDOI{\tempurl}


\bibitem[\protect\citeauthoryear{K{\"{o}}nighofer, Bloem, Ehlers, and
  Pek}{K{\"{o}}nighofer et~al\mbox{.}}{2022}]%
        {DBLP:conf/birthday/KonighoferBEP22}
\bibfield{author}{\bibinfo{person}{Bettina K{\"{o}}nighofer},
  \bibinfo{person}{Roderick Bloem}, \bibinfo{person}{R{\"{u}}diger Ehlers},
  {and} \bibinfo{person}{Christian Pek}.} \bibinfo{year}{2022}\natexlab{}.
\newblock \showarticletitle{Correct-by-Construction Runtime Enforcement in {AI}
  - {A} Survey}. In \bibinfo{booktitle}{\emph{Principles of Systems Design -
  Essays Dedicated to Thomas A. Henzinger on the Occasion of His 60th
  Birthday}} \emph{(\bibinfo{series}{LNCS}, Vol.~\bibinfo{volume}{13660})}.
  \bibinfo{publisher}{Springer}, \bibinfo{pages}{650--663}.
\newblock
\urldef\tempurl%
\url{https://doi.org/10.1007/978-3-031-22337-2\_31}
\showDOI{\tempurl}


\bibitem[\protect\citeauthoryear{Krasowski, Thumm, M{\"{u}}ller, Sch{\"{a}}fer,
  Wang, and Althoff}{Krasowski et~al\mbox{.}}{2023}]%
        {DBLP:journals/tmlr/KrasowskiTM0WA23}
\bibfield{author}{\bibinfo{person}{Hanna Krasowski}, \bibinfo{person}{Jakob
  Thumm}, \bibinfo{person}{Marlon M{\"{u}}ller}, \bibinfo{person}{Lukas
  Sch{\"{a}}fer}, \bibinfo{person}{Xiao Wang}, {and} \bibinfo{person}{Matthias
  Althoff}.} \bibinfo{year}{2023}\natexlab{}.
\newblock \showarticletitle{Provably Safe Reinforcement Learning: Conceptual
  Analysis, Survey, and Benchmarking}.
\newblock \bibinfo{journal}{\emph{Trans. Mach. Learn. Res.}}
  \bibinfo{volume}{2023} (\bibinfo{year}{2023}).
\newblock
\urldef\tempurl%
\url{https://openreview.net/forum?id=mcN0ezbnzO}
\showURL{%
\tempurl}


\bibitem[\protect\citeauthoryear{Lamport}{Lamport}{1977}]%
        {Lamport77}
\bibfield{author}{\bibinfo{person}{Leslie Lamport}.}
  \bibinfo{year}{1977}\natexlab{}.
\newblock \showarticletitle{Proving the Correctness of Multiprocess Programs}.
\newblock \bibinfo{journal}{\emph{{IEEE} Trans. Software Eng.}}
  \bibinfo{volume}{3}, \bibinfo{number}{2} (\bibinfo{year}{1977}),
  \bibinfo{pages}{125--143}.
\newblock
\urldef\tempurl%
\url{https://doi.org/10.1109/TSE.1977.229904}
\showDOI{\tempurl}


\bibitem[\protect\citeauthoryear{Larsen, Mikucionis, and Taankvist}{Larsen
  et~al\mbox{.}}{2015}]%
        {DBLP:conf/birthday/LarsenMT15}
\bibfield{author}{\bibinfo{person}{Kim~Guldstrand Larsen},
  \bibinfo{person}{Marius Mikucionis}, {and} \bibinfo{person}{Jakob~Haahr
  Taankvist}.} \bibinfo{year}{2015}\natexlab{}.
\newblock \showarticletitle{Safe and Optimal Adaptive Cruise Control}. In
  \bibinfo{booktitle}{\emph{Correct System Design - Symposium in Honor of
  Ernst-R{\"{u}}diger Olderog on the Occasion of His 60th Birthday}}
  \emph{(\bibinfo{series}{LNCS}, Vol.~\bibinfo{volume}{9360})}.
  \bibinfo{publisher}{Springer}, \bibinfo{pages}{260--277}.
\newblock
\urldef\tempurl%
\url{https://doi.org/10.1007/978-3-319-23506-6\_17}
\showDOI{\tempurl}


\bibitem[\protect\citeauthoryear{Li, Zhu, Braught, Shen, and Mitra}{Li
  et~al\mbox{.}}{2023}]%
        {DBLP:conf/cav/LiZBSM23}
\bibfield{author}{\bibinfo{person}{Yangge Li}, \bibinfo{person}{Haoqing Zhu},
  \bibinfo{person}{Katherine Braught}, \bibinfo{person}{Keyi Shen}, {and}
  \bibinfo{person}{Sayan Mitra}.} \bibinfo{year}{2023}\natexlab{}.
\newblock \showarticletitle{Verse: {A} Python Library for Reasoning About
  Multi-agent Hybrid System Scenarios}. In \bibinfo{booktitle}{\emph{{CAV}}}
  \emph{(\bibinfo{series}{LNCS}, Vol.~\bibinfo{volume}{13964})}.
  \bibinfo{publisher}{Springer}, \bibinfo{pages}{351--364}.
\newblock
\urldef\tempurl%
\url{https://doi.org/10.1007/978-3-031-37706-8\_18}
\showDOI{\tempurl}


\bibitem[\protect\citeauthoryear{Mikulski, Jamroga, and Kurpiewski}{Mikulski
  et~al\mbox{.}}{2022}]%
        {DBLP:conf/prima/MikulskiJK22}
\bibfield{author}{\bibinfo{person}{Lukasz Mikulski}, \bibinfo{person}{Wojciech
  Jamroga}, {and} \bibinfo{person}{Damian Kurpiewski}.}
  \bibinfo{year}{2022}\natexlab{}.
\newblock \showarticletitle{Assume-Guarantee Verification of Strategic
  Ability}. In \bibinfo{booktitle}{\emph{{PRIMA}}}
  \emph{(\bibinfo{series}{LNCS}, Vol.~\bibinfo{volume}{13753})}.
  \bibinfo{publisher}{Springer}, \bibinfo{pages}{173--191}.
\newblock
\urldef\tempurl%
\url{https://doi.org/10.1007/978-3-031-21203-1_11}
\showDOI{\tempurl}


\bibitem[\protect\citeauthoryear{Mnih, Kavukcuoglu, Silver, Rusu, Veness,
  Bellemare, Graves, Riedmiller, Fidjeland, Ostrovski, Petersen, Beattie,
  Sadik, Antonoglou, King, Kumaran, Wierstra, Legg, and Hassabis}{Mnih
  et~al\mbox{.}}{2015}]%
        {DBLP:journals/nature/MnihKSRVBGRFOPB15}
\bibfield{author}{\bibinfo{person}{Volodymyr Mnih}, \bibinfo{person}{Koray
  Kavukcuoglu}, \bibinfo{person}{David Silver}, \bibinfo{person}{Andrei~A.
  Rusu}, \bibinfo{person}{Joel Veness}, \bibinfo{person}{Marc~G. Bellemare},
  \bibinfo{person}{Alex Graves}, \bibinfo{person}{Martin~A. Riedmiller},
  \bibinfo{person}{Andreas Fidjeland}, \bibinfo{person}{Georg Ostrovski},
  \bibinfo{person}{Stig Petersen}, \bibinfo{person}{Charles Beattie},
  \bibinfo{person}{Amir Sadik}, \bibinfo{person}{Ioannis Antonoglou},
  \bibinfo{person}{Helen King}, \bibinfo{person}{Dharshan Kumaran},
  \bibinfo{person}{Daan Wierstra}, \bibinfo{person}{Shane Legg}, {and}
  \bibinfo{person}{Demis Hassabis}.} \bibinfo{year}{2015}\natexlab{}.
\newblock \showarticletitle{Human-level control through deep reinforcement
  learning}.
\newblock \bibinfo{journal}{\emph{Nat.}} \bibinfo{volume}{518},
  \bibinfo{number}{7540} (\bibinfo{year}{2015}), \bibinfo{pages}{529--533}.
\newblock
\urldef\tempurl%
\url{https://doi.org/10.1038/NATURE14236}
\showDOI{\tempurl}


\bibitem[\protect\citeauthoryear{Owicki and Gries}{Owicki and Gries}{1976}]%
        {OwickiG76}
\bibfield{author}{\bibinfo{person}{Susan~S. Owicki} {and}
  \bibinfo{person}{David Gries}.} \bibinfo{year}{1976}\natexlab{}.
\newblock \showarticletitle{Verifying Properties of Parallel Programs: An
  Axiomatic Approach}.
\newblock \bibinfo{journal}{\emph{Commun. {ACM}}} \bibinfo{volume}{19},
  \bibinfo{number}{5} (\bibinfo{year}{1976}), \bibinfo{pages}{279--285}.
\newblock
\urldef\tempurl%
\url{https://doi.org/10.1145/360051.360224}
\showDOI{\tempurl}


\bibitem[\protect\citeauthoryear{Partovi and Lin}{Partovi and Lin}{2014}]%
        {DBLP:conf/amcc/PartoviL14}
\bibfield{author}{\bibinfo{person}{Alireza Partovi} {and} \bibinfo{person}{Hai
  Lin}.} \bibinfo{year}{2014}\natexlab{}.
\newblock \showarticletitle{Assume-guarantee cooperative satisfaction of
  multi-agent systems}. In \bibinfo{booktitle}{\emph{{ACC}}}.
  \bibinfo{publisher}{{IEEE}}, \bibinfo{pages}{2053--2058}.
\newblock
\urldef\tempurl%
\url{https://doi.org/10.1109/ACC.2014.6859441}
\showDOI{\tempurl}


\bibitem[\protect\citeauthoryear{Pnueli}{Pnueli}{1984}]%
        {Pnueli84}
\bibfield{author}{\bibinfo{person}{Amir Pnueli}.}
  \bibinfo{year}{1984}\natexlab{}.
\newblock \showarticletitle{In Transition From Global to Modular Temporal
  Reasoning about Programs}. In \bibinfo{booktitle}{\emph{Logics and Models of
  Concurrent Systems}} \emph{(\bibinfo{series}{{NATO} {ASI} Series},
  Vol.~\bibinfo{volume}{13})}. \bibinfo{publisher}{Springer},
  \bibinfo{pages}{123--144}.
\newblock
\urldef\tempurl%
\url{https://doi.org/10.1007/978-3-642-82453-1_5}
\showDOI{\tempurl}


\bibitem[\protect\citeauthoryear{Puterman}{Puterman}{1994}]%
        {DBLP:books/wi/Puterman94}
\bibfield{author}{\bibinfo{person}{Martin~L. Puterman}.}
  \bibinfo{year}{1994}\natexlab{}.
\newblock \bibinfo{booktitle}{\emph{Markov Decision Processes: Discrete
  Stochastic Dynamic Programming}}.
\newblock \bibinfo{publisher}{Wiley}.
\newblock
\showISBNx{978-0-47161977-2}
\urldef\tempurl%
\url{https://doi.org/10.1002/9780470316887}
\showDOI{\tempurl}


\bibitem[\protect\citeauthoryear{Qin, Zhang, Chen, Chen, and Fan}{Qin
  et~al\mbox{.}}{2021}]%
        {DBLP:conf/iclr/QinZCCF21}
\bibfield{author}{\bibinfo{person}{Zengyi Qin}, \bibinfo{person}{Kaiqing
  Zhang}, \bibinfo{person}{Yuxiao Chen}, \bibinfo{person}{Jingkai Chen}, {and}
  \bibinfo{person}{Chuchu Fan}.} \bibinfo{year}{2021}\natexlab{}.
\newblock \showarticletitle{Learning Safe Multi-agent Control with
  Decentralized Neural Barrier Certificates}. In
  \bibinfo{booktitle}{\emph{{ICLR}}}. \bibinfo{publisher}{OpenReview.net}.
\newblock
\urldef\tempurl%
\url{https://openreview.net/forum?id=P6\_q1BRxY8Q}
\showURL{%
\tempurl}


\bibitem[\protect\citeauthoryear{Raju, Bharadwaj, Djeumou, and Topcu}{Raju
  et~al\mbox{.}}{2021}]%
        {RajuBDT21}
\bibfield{author}{\bibinfo{person}{Dhananjay Raju},
  \bibinfo{person}{Sudarshanan Bharadwaj}, \bibinfo{person}{Franck Djeumou},
  {and} \bibinfo{person}{Ufuk Topcu}.} \bibinfo{year}{2021}\natexlab{}.
\newblock \showarticletitle{Online Synthesis for Runtime Enforcement of Safety
  in Multiagent Systems}.
\newblock \bibinfo{journal}{\emph{{IEEE} Trans. Control. Netw. Syst.}}
  \bibinfo{volume}{8}, \bibinfo{number}{2} (\bibinfo{year}{2021}),
  \bibinfo{pages}{621--632}.
\newblock
\urldef\tempurl%
\url{https://doi.org/10.1109/TCNS.2021.3061900}
\showDOI{\tempurl}


\bibitem[\protect\citeauthoryear{Schulman, Wolski, Dhariwal, Radford, and
  Klimov}{Schulman et~al\mbox{.}}{2017}]%
        {DBLP:journals/corr/SchulmanWDRK17}
\bibfield{author}{\bibinfo{person}{John Schulman}, \bibinfo{person}{Filip
  Wolski}, \bibinfo{person}{Prafulla Dhariwal}, \bibinfo{person}{Alec Radford},
  {and} \bibinfo{person}{Oleg Klimov}.} \bibinfo{year}{2017}\natexlab{}.
\newblock \showarticletitle{Proximal Policy Optimization Algorithms}.
\newblock \bibinfo{journal}{\emph{CoRR}} (\bibinfo{year}{2017}).
\newblock
\showeprint[arXiv]{1707.06347}


\bibitem[\protect\citeauthoryear{Stark}{Stark}{1985}]%
        {Stark85}
\bibfield{author}{\bibinfo{person}{Eugene~W. Stark}.}
  \bibinfo{year}{1985}\natexlab{}.
\newblock \showarticletitle{A Proof Technique for Rely/Guarantee Properties}.
  In \bibinfo{booktitle}{\emph{{FSTTCS}}} \emph{(\bibinfo{series}{LNCS},
  Vol.~\bibinfo{volume}{206})}. \bibinfo{publisher}{Springer},
  \bibinfo{pages}{369--391}.
\newblock
\urldef\tempurl%
\url{https://doi.org/10.1007/3-540-16042-6_21}
\showDOI{\tempurl}


\bibitem[\protect\citeauthoryear{Sutton and Barto}{Sutton and Barto}{1998}]%
        {DBLP:books/lib/SuttonB98}
\bibfield{author}{\bibinfo{person}{Richard~S. Sutton} {and}
  \bibinfo{person}{Andrew~G. Barto}.} \bibinfo{year}{1998}\natexlab{}.
\newblock \bibinfo{booktitle}{\emph{Reinforcement Learning: An Introduction}}.
\newblock \bibinfo{publisher}{{MIT} Press}.
\newblock
\showISBNx{978-0-262-19398-6}
\urldef\tempurl%
\url{http://incompleteideas.net/book/the-book-1st.html}
\showURL{%
\tempurl}


\bibitem[\protect\citeauthoryear{Xiao, Lyu, and Dolan}{Xiao
  et~al\mbox{.}}{2023}]%
        {DBLP:conf/atal/XiaoLD23}
\bibfield{author}{\bibinfo{person}{Wenli Xiao}, \bibinfo{person}{Yiwei Lyu},
  {and} \bibinfo{person}{John~M. Dolan}.} \bibinfo{year}{2023}\natexlab{}.
\newblock \showarticletitle{Model-based Dynamic Shielding for Safe and
  Efficient Multi-agent Reinforcement Learning}. In
  \bibinfo{booktitle}{\emph{{AAMAS}}}. \bibinfo{publisher}{{ACM}},
  \bibinfo{pages}{1587--1596}.
\newblock
\urldef\tempurl%
\url{https://doi.org/10.5555/3545946.3598814}
\showDOI{\tempurl}


\bibitem[\protect\citeauthoryear{Yang, Marra, Rens, and Raedt}{Yang
  et~al\mbox{.}}{2023}]%
        {DBLP:conf/ijcai/YangMRR23}
\bibfield{author}{\bibinfo{person}{Wen{-}Chi Yang}, \bibinfo{person}{Giuseppe
  Marra}, \bibinfo{person}{Gavin Rens}, {and} \bibinfo{person}{Luc~De Raedt}.}
  \bibinfo{year}{2023}\natexlab{}.
\newblock \showarticletitle{Safe Reinforcement Learning via Probabilistic Logic
  Shields}. In \bibinfo{booktitle}{\emph{{IJCAI}}}.
  \bibinfo{publisher}{ijcai.org}, \bibinfo{pages}{5739--5749}.
\newblock
\urldef\tempurl%
\url{https://doi.org/10.24963/IJCAI.2023/637}
\showDOI{\tempurl}


\bibitem[\protect\citeauthoryear{Yu, Velu, Vinitsky, Gao, Wang, Bayen, and
  Wu}{Yu et~al\mbox{.}}{2022}]%
        {DBLP:conf/nips/YuVVGWBW22}
\bibfield{author}{\bibinfo{person}{Chao Yu}, \bibinfo{person}{Akash Velu},
  \bibinfo{person}{Eugene Vinitsky}, \bibinfo{person}{Jiaxuan Gao},
  \bibinfo{person}{Yu Wang}, \bibinfo{person}{Alexandre~M. Bayen}, {and}
  \bibinfo{person}{Yi Wu}.} \bibinfo{year}{2022}\natexlab{}.
\newblock \showarticletitle{The Surprising Effectiveness of {PPO} in
  Cooperative Multi-Agent Games}. In \bibinfo{booktitle}{\emph{{NeurIPS}}}.
\newblock
\urldef\tempurl%
\url{http://papers.nips.cc/paper\_files/paper/2022/hash/9c1535a02f0ce079433344e14d910597-Abstract-Datasets\_and\_Benchmarks.html}
\showURL{%
\tempurl}


\bibitem[\protect\citeauthoryear{Zhang, Yang, and Basar}{Zhang
  et~al\mbox{.}}{2019}]%
        {DBLP:journals/corr/ZhangYB19}
\bibfield{author}{\bibinfo{person}{Kaiqing Zhang}, \bibinfo{person}{Zhuoran
  Yang}, {and} \bibinfo{person}{Tamer Basar}.} \bibinfo{year}{2019}\natexlab{}.
\newblock \showarticletitle{Multi-Agent Reinforcement Learning: {A} Selective
  Overview of Theories and Algorithms}.
\newblock \bibinfo{journal}{\emph{CoRR}} (\bibinfo{year}{2019}).
\newblock
\showeprint[arXiv]{1911.10635}


\bibitem[\protect\citeauthoryear{Zhang and Bastani}{Zhang and Bastani}{2019}]%
        {DBLP:journals/corr/ZhangB19}
\bibfield{author}{\bibinfo{person}{Wenbo Zhang} {and} \bibinfo{person}{Osbert
  Bastani}.} \bibinfo{year}{2019}\natexlab{}.
\newblock \showarticletitle{{MAMPS:} Safe Multi-Agent Reinforcement Learning
  via Model Predictive Shielding}.
\newblock \bibinfo{journal}{\emph{CoRR}} (\bibinfo{year}{2019}).
\newblock
\showeprint[arXiv]{1910.12639}


\end{thebibliography}

\ifshowappendix

\clearpage
\appendix

\section{Appendix}

\subsection{Policy of the Environment-Controlled Car}\label{sect:frontcarspeed}
The environment-controlled front car decides between accelerations of respectively $-2$\,m/s$^2$, $0$\,m/s$^2$, or $2$\,m/s$^2$ through a random weighted draw. The weights that are used for the draw ($w_{-2}, w_0, w_2$) are influenced by the environment-controlled car's own velocity, $v_n$, in the following manner:

\begin{align*}
  w_{-2} &= \begin{cases}
    2 \text{ if  } v_n > 10\\
    1 \text{ otherwise}
  \end{cases}\\
  w_0 &= 1 \\
  w_{2} &= \begin{cases}
    2 \text{ if  } v_n < 0\\
    1 \text{ otherwise}
  \end{cases}
\end{align*}

\subsection{Demand and Cost Patterns of the Chemical Production Plant}\label{sect:plant}

\begin{figure}[H]
\begin{subfigure}{\linewidth}
  \centering
  \includesvg[width=\linewidth]{"Graphics/Consumer_Demand.svg"}
  \caption{Periodically varying demand by consumers.}
  \Description{Two graphs showing consumption as a function of time. Consumption changes once per second. Values given in this description are approximate. The first graph shows that the consumption from consumer A follows a cycle of 5, 5, 3, 0, then 0. The second graph shows that the consumption from consumer B has a cycle of 5, 3, 3, 3, then 0. }
  \label{fig:consumerdemand}
\end{subfigure}
\begin{subfigure}{\linewidth}
  \centering
  \includesvg[width=\linewidth]{"Graphics/Provider_Cost.svg"}
  \caption{Exemplary periodically varying cost of the providers for units~1 and~10.  When there are multiple providers to the same unit, they all have the same cost.}
  \Description{Two graphs showing cost as a function of time. Cost changes once per second. Values given in this description are approximate. The first graph shows that the cost from provider 1 is 0, 5, 3, 3, then 3. The second shows that the cost from provider 10 is 9, 0, 6, 6, then 6.}
  \label{fig:providercost}
\end{subfigure}
\caption{Patterns from the chemical production plant.}
\label{fig:patterns_demandcost}
\end{figure}

\subsection{Chemical Production Plant: MAPPO Safety}\label{sect:cpmapposfaety}

The agents controlling the chemical production units were penalized by an immediate cost of $25{,}600$ whenever they were in unsafe states.
We arrived at that penalty value by the same process as the car platoon example, i.e., starting from $100$ and doubling the penalty until the rate of safety started to diminish.
The severity of this penalty means that a single highly unsafe outlier can skew the mean performance massively, creating the spikes of bad performance seen in Figure~\ref{fig:cplearning}.

Figure~\ref{fig:cpmappopercentagesafe} shows the resulting fraction of safe runs, learned under this penalty.

\begin{figure}[H]
  \centering
  \includesvg[width=\linewidth]{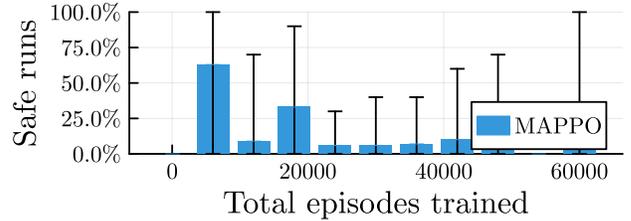}
  \caption{Percentage of safe runs with the MAPPO policy in the chemical production example. Blue bars show the mean of 10 repetitions, while black intervals give min and max values.}
  \Description{Bar chart. After 60,000 training episodes, the best training outcome achieves 100 percent safe runs. However, the worst outcome for the same number of episodes showed 0 percent, and the mean is less than 25 percent. The general level of safety is very low, with mean safety staying well below 25 percent.}
  \label{fig:cpmappopercentagesafe}
\end{figure}

\fi

\end{document}